\documentclass[final,nonblindrev]{informs_modified} 

\OneAndAHalfSpacedXI 

\usepackage{graphicx}
\usepackage{subfigure, epsfig}
\usepackage{natbib}
\usepackage{booktabs}
 \bibpunct[, ]{(}{)}{,}{a}{}{,}%
 \usepackage[dvipsnames]{xcolor}
\usepackage{algorithm}
\usepackage[noend]{algorithmic}
\usepackage{setspace}
\usepackage{bm}
\usepackage{tabu}
\usepackage{array}
\usepackage{multirow}
\usepackage{float}
\usepackage{comment}
\usepackage[toc,page]{appendix}
\usepackage{appendix}
\usepackage{setspace}
\usepackage{hyperref}
\usepackage{enumitem}
\usepackage{mathtools}

\TheoremsNumberedThrough     
\ECRepeatTheorems

\EquationsNumberedThrough    

\newcommand{\xhdr}[1]{\vspace{1mm} \noindent{\bf #1}}

\begin{document}


\RUNAUTHOR{Yao and Zhou}

\RUNTITLE{Fair LLM Serving with Balanced Batch Constraints}


\TITLE{Efficiency and Cost Alignment in Batched LLM Serving via Resource-Fair Scheduling}

\ARTICLEAUTHORS{%
\AUTHOR{Dayi Yao}
\AFF{University of Washington, 
\EMAIL{ydy@uw.edu}} 
\AUTHOR{Zijie Zhou}
\AFF{IEDA, HKUST, 
\EMAIL{jerryzhou@ust.hk}}
} 

\ABSTRACT{%
This paper studies a resource-allocation inefficiency in batched large language model (LLM) serving: heterogeneous requests that share a decode batch impose max-driven computational costs on one another. Because the wall-clock cost of a batch step is largely governed by the largest active KV-cache footprint, a short request co-batched with a long request can experience latency and GPU-resource consumption disproportionate to its own token workload. We formalize this phenomenon as a resource-fair scheduling problem.

We develop a mathematical scheduling model that connects within-batch resource fairness to system throughput. The proposed fairness constraint bounds the disparity in decode progress, equivalently KV-cache footprint, among co-batched requests. Based on this model, we design the Insert-Short-Jobs-with-Limit (ISJL) algorithm, a parameterized hybrid batching policy. We prove that ISJL achieves a global competitive-ratio lower bound of $3/4$ and further characterize how the worst-case guarantee varies with the fairness budget. Empirical evaluations show that ISJL improves throughput and end-to-end latency relative to standard scheduling baselines.

We further examine the profit implications of resource-fair scheduling under the token-metered pricing convention used by commercial LLM APIs. In this environment, for a fixed accepted workload, revenue is additive in token counts and is independent of the scheduling policy; scheduling affects profit through inference cost. We prove a cost decomposition showing that inference cost consists of an intrinsic workload term, a time/overhead term, and a schedule-induced batching externality term. The fairness constraint bounds the batching externality, while the ISJL throughput guarantee controls the time/overhead component. Numerical experiments show that ISJL occupies a favorable middle ground between FCFS, which has large batching externalities, and LJF, which is cost-aligned but sacrifices batching flexibility. Thus, ISJL provides a bi-criterion scheduling policy: it maintains high throughput while aligning max-driven batch cost with token-metered revenue.}
%

\KEYWORDS{scheduling, computational fairness, LLM serving}

\maketitle

%


\section{Introduction} \label{sec:intro}

Large language models (LLMs) \citep{brown2020language,chowdhery2023palm,openai2023gpt} now underpin a wide range of information \citep{anthropic2023claude,chatgpt2023,openai2023gpt}, assistants \citep{codewhisperer2023,githubcopilot2023}, and decision-support applications \citep{huang2025orlm,cascella2023evaluating,sallam2023utility}. LLM inference is the dominant, recurring cost for operators: requests arrive online and are served on GPU workers that exploit batching for throughput. The mechanics are distinctive. Each request is processed in two phases \citep{kwon2023efficient}: a \emph{prefill} pass over the input tokens, followed by an \emph{autoregressive decode} phase that generates output one token at a time. During decode, each step reuses the key–value (KV) cache from previously generated tokens; the compute and memory required per step therefore increase with a request’s accumulated context length. Modern systems \citep{kwon2023efficient,agrawal2023sarathi} explicitly manage KV memory and construct batches across heterogeneous requests to raise utilization, making batching a first-order operational lever at inference time. 

This batching lever creates a subtle externality. When a very long request and a very short request decode in the same batch, the short request “pays” latency to the long one: the per-step FLOPs and memory footprint are driven by the longest context in the batch. In effect, short jobs subsidize long jobs through shared compute and KV memory. The phenomenon is well documented in \cite{sheng2024fairness,khan2024ensuring,wei2025equinox} that measures lower GPU utilization in decode, sensitivity to sequence length, and the benefits of uniform micro-batches, but it has not been formalized as a \emph{resource fairness} problem in service operations. We take that step: rather than defining fairness through outcomes or service guarantees, we define \textbf{fairness} directly in terms of \emph{consumed computational resources} within a batch—proxied by each request’s current decode progress (i.e., number of generated tokens and associated KV cache).

\xhdr{Resource fairness as an operational batching discipline.}
The fairness notion we study is distinct from the ethical or 
equity-based fairness concepts prevalent in operations management 
and algorithm design (e.g., group or individual parity in predictive 
systems, egalitarian welfare). We study \emph{resource-consumption fairness}: 
a batching discipline requiring co-batched requests to have comparable decode 
progress, or equivalently, comparable KV-cache footprints. The purpose of this 
discipline is to limit the extent to which a request with a small resource 
footprint is served in a batch whose resource profile is determined by much 
larger requests.

This perspective leads to a cost--efficiency frontier. Tighter resource 
homogeneity improves cost alignment by reducing extreme within-batch 
heterogeneity, but it can also restrict batching flexibility. Our model 
therefore treats $\alpha$ as a resource-homogeneity budget and asks how much 
throughput can be retained while respecting that budget. This differs from the 
classical fairness--efficiency tradeoff in service operations: fairness here is 
not an outcome-parity requirement, but an operational control that limits the 
batching externality described above.

\xhdr{Resource-fair scheduling improves cost alignment under token-metered pricing.}
A second managerial implication concerns the profit consequences of scheduling under prevailing API pricing practice. Commercial LLM providers typically use token-metered contracts: model-specific rates are quoted for input, output, and sometimes cached tokens, usually per one million tokens \citep{openai_api_pricing,anthropic_api_pricing,deepseek_api_pricing,google_gemini_api_pricing}. These contracts are not decisions made by the scheduler; they form an external pricing environment. Consistent with the token-level abstraction in Section~\ref{sec:model}, we represent this environment by a single effective linear charge $p_i=p_0 o_i$ for request $i$. Once a set of requests has been accepted, total revenue $\sum_i p_i$ is fixed and additive in token counts. Scheduling affects profit through inference cost.

This creates an operational tension. Token-metered revenue is additive across requests, whereas the per-step cost of a batch is driven by the largest active resource footprint in the batch. Heterogeneous co-batching can therefore create a schedule-induced excess cost: short requests may be processed in steps whose cost is determined by longer co-batched requests. Resource-fair scheduling addresses this cost-alignment problem. By requiring co-batched requests to have comparable decode progress, or equivalently comparable KV-cache footprints, the scheduler bounds the gap between max-driven batch cost and the additive token workload measured by pricing. Thus, our economic analysis is not a pricing-design model. Rather, it evaluates the profit and cost-alignment implications of resource-fair scheduling under exogenous token-metered prices.

\subsection{Main Contributions}

\xhdr{A formal model linking fairness and throughput in LLM inference. } In Section~\ref{sec:model}, we develop a single-worker scheduling model that explicitly captures the dynamics of LLM inference. The model assumes that a decision-maker receives $n$ heterogeneous requests, each requiring sequential token generation, and that the worker can process at most $B$ tokens concurrently. To maximize throughput, the worker naturally seeks to form full batches of size $B$ at each step. However, such batching can create substantial unfairness in resource consumption. To address this inefficiency, we introduce a parametric fairness constraint indexed by $\alpha$, which bounds the disparity in resource consumption (measured by decoding progress or KV-cache footprint) among co-batched requests. This model enables a quantitative characterization of the tradeoff between throughput and fairness, and provides a tractable foundation for algorithmic design.

\xhdr{Fair and efficient scheduling algorithms with provable guarantees. } In Section~\ref{sec:alg}, we design and analyze algorithms that achieve both fairness and high throughput under the proposed constraint. We begin with a baseline policy—\emph{Longest-Job-First (LJF)}—which is universally fair for any $\alpha$. We prove in Theorem~\ref{thm:ljf-competitive} that LJF attains a competitive ratio of at most $\frac{B}{2B-1}$. We then introduce a refined parameterized algorithm, \emph{Insert-Short-Jobs-with-Limit (ISJL)}, which explicitly incorporates the fairness parameter $\alpha$. Theorems~\ref{thm:cr-isj-limited} and~\ref{thm:cr-isj-lim-gen} establish that ISJL achieves a competitive ratio of at least $0.77$ when $B=2$, and a uniform lower bound of $3/4$ for all $B \ge 2$. These results provide the first approximation guarantees for resource-fair scheduling in LLM inference settings. {In Section~\ref{sec:robust}, 
we further extend ISJL to the practically important 
setting where output token lengths are unknown at 
arrival time. The proposed \emph{Robust-ISJL} 
algorithm enforces the fairness constraint using 
only observed runtime progress, requiring no 
knowledge of true output lengths, and retains a 
competitive guarantee that degrades gracefully with 
prediction error, recovering the $3/4$ bound exactly 
when predictions are perfect.}

\xhdr{Competitive Ratio Lower Bound as a Function of $\alpha$}
In Section~\ref{sec:extension}, we sharpen the guarantees by characterizing how the worst-case competitive ratio varies with the fairness budget.
Let $\gamma=\alpha/o_1\in(0,1)$ denote the normalized fairness parameter (relative to the longest job).
For all $B\ge2$, the overall worst-case lower bound takes the single piecewise form
\[
\mathrm{CR}(\gamma)=
\begin{cases}
\dfrac{1+\gamma}{1+2\gamma}, & 0<\gamma\le\tfrac12,\\[4pt]
\dfrac{2-\gamma}{3-2\gamma}, & \tfrac12\le\gamma<1,
\end{cases}
\]
which is strictly decreasing on $(0,\tfrac12]$, strictly increasing on $[\tfrac12,1)$, and attains a unique global minimum of $\tfrac34$ at $\gamma=\tfrac12$; moreover, $\mathrm{CR}(\gamma)\to1$ as $\gamma\to0^+$ or $\gamma\to1^-$. 
This parameterized bound, obtained by analyzing over- and under-inserting adversaries, both explains when the $\tfrac34$ limit is tight and clarifies the fairness–efficiency trade-off induced by $\alpha$.

\xhdr{Profit and cost alignment under exogenous token-metered pricing.}
Section~\ref{sec:profit-cost-alignment} connects the scheduling analysis to the token-metered pricing convention used in commercial LLM APIs. We treat pricing as exogenous and use profit as an evaluation dimension of scheduling, rather than as a pricing decision. For a fixed accepted workload, linear token revenue is independent of the schedule, so profit differences across schedulers arise from inference cost. We prove a cost decomposition showing that the total inference cost of a schedule can be written as
\[
C_{\mathcal A}(\mathcal I)
=
\tau Q(\mathcal I)+cT_{\mathcal A}(\mathcal I)+\tau E_{\mathcal A}(\mathcal I),
\]
where $Q(\mathcal I)$ is an intrinsic workload term independent of scheduling, $T_{\mathcal A}(\mathcal I)$ is the makespan, and $E_{\mathcal A}(\mathcal I)$ is the schedule-induced batching externality generated by max-driven batch cost. We show that any $\alpha$-fair schedule satisfies $E_{\mathcal A}(\mathcal I)\le \alpha O(\mathcal I)$, where $O(\mathcal I)$ is the total token workload. Combining this externality bound with the throughput guarantee of ISJL yields an additive profit guarantee relative to the optimal $\alpha$-fair profit schedule. The numerical experiments then decompose profit into intrinsic workload cost, time/overhead cost, and batching-externality cost. They show that ISJL substantially reduces FCFS's externality cost while preserving much of the batching flexibility that LJF sacrifices. We also explicitly discuss the scope of the demand specification: customer acceptance is modeled as exogenous to realized waiting times, so the results quantify cost alignment under a fixed pricing environment rather than endogenous pricing or demand response.

\xhdr{System efficiency on real workloads. }
Section~\ref{sec:num} benchmarks schedulers on a fixed arrival sequence drawn from LMSYS-Chat-1M. Using throughput and Average End-to-end Latency (AEL) as metrics under $B\in\{16,32\}$, ISJL dominates these baselines, simultaneously increasing throughput and reducing AEL across tested $\alpha\in\{50,100,150\}$, with $\alpha$ offering a tunable efficiency–latency trade-off. These controlled results complement the revenue findings, showing ISJL delivers balanced performance with near top profit and markedly better efficiency on realistic, heterogeneous workloads. 

\subsection{Other Related Works}

\xhdr{LLM Inference and Fair Serving. } Enhancing the efficiency of LLM inference is essential, as it involves considerable computational and financial overhead. Most prior studies emphasize system-level engineering approaches to speed up inference in deployed environments, typically lacking theoretical performance assurances. For example, \citet{patel2023splitwise, zhong2024distserve} design architectures that decouple the processing of prompts and tokens, whereas \citet{yu2022orca, agrawal2023sarathi, agrawal2024taming} develop integrated designs that jointly process both—an architecture consistent with the setting analyzed in this work. Moreover, there are some recent works focusing on the service level fairness in LLM inference \citet{sheng2024fairness,khan2024ensuring}.

More recently, several studies have begun to establish theoretical underpinnings for LLM inference. The mathematical formulation adopted in this paper extends the framework of \citet{jaillet2025online}. \cite{wang2025llm,chen2025adaptively} improve the efficiency and robustness of this theoretical model. \citet{ao2025optimizing} investigate inference scheduling problems with multiple objectives and design algorithms that provide formal performance guarantees.

\xhdr{Scheduling Optimization. } The scheduling problem—widely explored in prior research \citep{allahverdi2008survey, chen1998review, albers2009online, mak2015appointment, kong2013scheduling}—concerns a decision-maker responsible for processing numerous incoming tasks, either sequentially or in groups. The key difficulty is to determine an efficient order and timing for execution, typically formalized as an integer optimization problem. In our setting, inspired by LLM inference, multiple jobs can be executed simultaneously in batches, connecting our formulation to batch scheduling frameworks studied in \citep{chen2008logistics, liu2015online, li2020online}. Moreover, LLM inference introduces an inherent sequential dependency—subsequent tokens cannot be processed until the current ones are completed—thereby creating precedence constraints. Related studies on scheduling under such constraints include \citep{shahout2024don, precSchedAnupam, schedPrecedence, precSchedSchabanel}. Nonetheless, these works differ fundamentally from ours, as our formulation uniquely incorporates a fairness constraint on resource usage arising specifically from the LLM inference context.

\section{Model} \label{sec:model}

In this paper, we study a scheduling problem for large language model (LLM) inference on a single computational worker. A total of $n$ prompt requests await processing. Each request $i \in [n]$ is characterized by a parameter $o_i$, where $o_i$ denotes the total number of tokens in the prompt input and corresponding model-generated output.

For each request, the chunked input technique from (\cite{agrawal2023sarathi}) processes tokens sequentially, one at a time. Once the first output token is generated, the remaining tokens follow in the same manner. Thus, each request completes within $o_i$ units of time, processing $o_i$ tokens sequentially. { Throughout the main analysis we assume $o_i$ is known to the scheduler at the 
time of admission, which is a standard assumption in the offline scheduling 
literature~\citep{jaillet2025online, wang2025llm}. In Section~\ref{sec:robust}, we remove 
this assumption and extend the model to the non-clairvoyant setting, in which 
only a prediction interval $[\ell_i, u_i]$ with $\ell_i\le o_i\le u_i$ is 
available. }

The computational worker processes requests in batches, with a maximum batch size of $B$ requests. Each batch may consist of tokens from different requests, but not multiple decoding tokens from the same request, as these are generated sequentially. After processing a batch, each request included in it advances by generating its next output token. The model uses a logical token-step abstraction: one service step advances each 
active request by one token. Thus, the makespan $T_{\mathcal A}(\mathcal I)$ 
counts the number of scheduling steps rather than calibrated GPU milliseconds. 
We adopt this abstraction deliberately. In deployed LLM serving systems, the 
wall-clock time of a batch step depends on hardware and system choices, and there is no universal 
closed-form latency function that is valid across serving architectures. The batching externality motivates the fairness 
constraint below, while the profit and cost-alignment implications of 
max-driven batch cost are developed separately in 
Section~\ref{sec:profit-cost-alignment}.

\xhdr{Fair Serving Metric. } To ensure equitable resource allocation, we impose a fairness constraint based on the number of tokens held in the key-value (KV) cache: At any time $t$, let $S^{(t)}$ denote the set of requests being processed, and let $a_i^{(t)} \in [0, o_i]$ denote the number of decoding tokens generated for request $i \in S^{(t)}$. The memory footprint of request $i$ at time $t$ is thus $a_i^{(t)}$. We require that:
\begin{equation} \label{eq:fairness}
\max_{i \in S^{(t)}} \{a_i^{(t)} \} - \min_{i \in S^{(t)}} \{a_i^{(t)} \} \leq \alpha
\end{equation}
for some fixed constant $\alpha$. 
{To motivate constraint~\eqref{eq:fairness}, it is important to understand why 
the computational cost of a batch step is driven by the \emph{longest} KV 
cache in the batch rather than by the sum or average.
At each decode step, two GPU operations dominate cost. First, the 
\emph{attention} operation for request $i$ scales with its current KV cache 
size $a_i^{(t)}$; however, since GPU cores can process attention heads in 
parallel, the wall-clock time for this operation is approximately proportional 
to $\max_{i\in S^{(t)}} a_i^{(t)}$, not the sum. Second, the \emph{feed-forward} 
layers compute a joint matrix product over all tokens in the batch, and 
the projection matrices are loaded from GPU high-bandwidth memory 
\emph{once} for the entire batch regardless of its size. Due to GPU
parallelism, the wall-clock time for this shared computation is approximately 
constant over a wide range of batch sizes \citep{le2023dissecting}. 
Both mechanisms imply that a short request co-batched with a long one pays 
latency proportional to the long request's KV footprint, not its own. This is 
the \emph{batching externality}: long jobs impose a negative external cost on 
short jobs by inflating the per-step compute time they experience. 

Constraint~\eqref{eq:fairness} directly bounds this resource heterogeneity. 
When the max--min progress gap is at most $\alpha$, all co-batched requests have 
KV caches within $\alpha$ tokens of each other, so the largest active footprint 
in the batch remains close to each request's own footprint. In this sense, 
$\alpha$ controls the schedule-induced gap between individual token progress 
and the resource profile of the batch. Section~\ref{sec:profit-cost-alignment} 
uses this observation to formalize the resulting cost-alignment implication 
under max-driven batch cost.}

{
\xhdr{Scheduling Policy.}
A \emph{scheduling policy} is an event-driven rule that, at each decision epoch $t$,
determines the active batch $S^{(t)}\subseteq[n]$ subject to: (i) the capacity
constraint $|S^{(t)}|\le B$, and (ii) the fairness constraint~\eqref{eq:fairness}.
Crucially, a policy is \emph{not} restricted to \emph{static batching}, in which
a fixed set of $B$ requests is locked together until every member finishes.
A policy may instead implement \emph{continuous batching}: whenever a request
completes and frees a slot, the policy may immediately admit a new request into the
active batch---provided the fairness constraint remains satisfied---without waiting
for any other running request to complete.
Continuous batching is the default operation mode of modern LLM serving
systems such as the First-Come-First-Serve (FCFS) scheduling framework in  Orca~\citep{yu2022orca} and vLLM~\citep{kwon2023efficient}, which will be introduced in Section \ref{subsec:example} in detail.}

{
\xhdr{Objective and Performance Metric.}
Our objective is to design a scheduling algorithm that 
maximizes \emph{throughput}, defined as the total number 
of tokens processed per unit time:
\[
\text{Throughput}(\mathcal{A};\mathcal{I}) 
= \frac{\sum_{i=1}^n o_i}{T_{\mathcal{A}}(\mathcal{I})},
\]
where $T_{\mathcal{A}}(\mathcal{I})$ is the makespan of algorithm $\mathcal{A}$ on 
instance $\mathcal{I}$, and $\sum_{i=1}^n o_i$ is 
fixed for a given instance.

Note that not every batch can process $B$ requests: 
when the fairness constraint~\eqref{eq:fairness} is 
active, it may not be possible to find $B$ requests 
whose current decode progress falls within a window 
of size $\alpha$, and the scheduler must form a 
smaller batch or admit only requests that satisfy 
the constraint. The throughput objective therefore 
captures the tension between batch fullness and 
fairness enforcement.

\begin{proposition}[NP-hardness]
\label{prop:np}
Assume $n > B$. Finding a schedule that maximizes 
throughput subject to the fairness 
constraint~\eqref{eq:fairness} is NP-hard.
\end{proposition}

Proposition \ref{prop:np} follows straightforwardly by setting $\alpha=\max_{i \in [n]} o_i$. With this choice, the fairness constraint in Equation \eqref{eq:fairness} becomes trivially satisfied for all possible schedules. In the absence of fairness constraints, prior work by \cite{lee2001machine} establishes that throughput maximization is strongly NP-hard. Consequently, the problem necessitates the design of approximation algorithms.

To evaluate the performance of an approximation algorithm, we use the 
\emph{competitive ratio}. Given a request instance 
$\mathcal{I}$, let $\text{OPT}(\mathcal{I})$ denote 
the optimal throughput achievable by any fair schedule 
(obtainable via the integer program in 
Appendix~\ref{append:model}). The competitive ratio 
of algorithm $\mathcal{A}$ is defined as:
\[
CR(\mathcal{A}) = \inf_{\mathcal{I}} 
\frac{\text{Throughput}(\mathcal{A};\mathcal{I})}
{\text{OPT}(\mathcal{I})}.
\]
A competitive ratio of $1$ means the algorithm 
always achieves optimal throughput; a ratio of 
$c < 1$ means the algorithm's throughput is 
guaranteed to be at least a $c$ fraction of 
optimal on every instance.

}

\xhdr{Bicriteria interpretation of the constrained objective.}
The formulation above should be interpreted as a constrained bicriteria
optimization problem. The operator cares about both cost alignment and system
efficiency. The parameter $\alpha$ controls the first dimension by imposing a
resource-homogeneity budget on each active batch; conditional on this budget,
the scheduler optimizes the second dimension, namely throughput, or equivalently
makespan. This constraint-and-optimize formulation is common in multi-objective
online decision problems: rather than specifying a single scalar weight across
objectives, one objective is parameterized as a constraint and the other is
optimized. A related philosophy appears in the consistency--robustness
literature for online algorithms with predictions, where a tunable parameter is
used to trace a Pareto frontier between two performance notions
\citep{lykouris2021competitive,golrezaei2023online}.

\begin{remark}[Offline snapshot model and online arrivals]
\label{rem:offline-online-scope}
The main theoretical model is an offline backlog model: all requests in
$\mathcal I$ are available at time $0$. This formulation should be interpreted
as a snapshot approximation of a congested serving system, in which the operator
periodically faces a large queue of pending requests and must prioritize them
under a batch-capacity constraint. 
This model does not claim to solve the fully online arrival problem. In a
production system, new requests continue to arrive while the current backlog is
being served, and the scheduler must make causal decisions without knowing
future arrivals. This introduces a distinct information-structure problem:
whether to serve currently available requests immediately or wait for future
requests that may form better resource-homogeneous batches. The competitive-ratio
results in Sections~\ref{subsec:LJF}--\ref{sec:robust} are therefore
offline guarantees for backlogged instances. Section~\ref{sec:online-discussion}
discusses how to adapt our offline algorithms under online arrivals and why the
offline competitive guarantee does not directly extend to adversarial release times.
\end{remark}

{
\begin{remark}[Compute-bound vs.\ memory-bound constraints in LLM inference]
\label{rem:compute-memory}
LLM inference systems face two distinct resource constraints that can 
each become the binding bottleneck depending on hardware configuration 
and workload characteristics.

The \emph{memory-bound} constraint limits the total number of KV-cache 
tokens that can be held in GPU high-bandwidth memory (HBM) simultaneously. 
When sequences are very long, GPU HBM fills up before 
compute is saturated, forcing the system to pause or preempt active 
requests. This constraint is the focus of recent theoretical work by 
\citet{jaillet2025online}, and is also the engineering motivation for 
systems such as PagedAttention~\citep{kwon2023efficient}, which manages 
KV memory via virtual paging to delay the onset of memory-bound behavior.

The \emph{compute-bound} constraint limits the number of requests that 
can be processed concurrently by GPU compute throughput. When sequences 
are short-to-medium in length, the system can hold all active KV caches in 
memory simultaneously, and the binding constraint is the number of 
requests $B$ that can be batched in a single compute step. This is the 
formulation adopted by Orca~\citep{yu2022orca}, 
Sarathi-Serve~\citep{agrawal2024taming}, and 
Sheng et al.~\citep{sheng2024fairness}, all of which use a maximum 
concurrent sequence count as the primary scheduling constraint. It is 
also the regime described by the GPU batching analysis 
in~\citet{le2023dissecting}: the observation that MLP compute time is 
approximately constant across a wide range of batch sizes holds precisely 
because in the compute-bound regime, the GPU is underutilized at small 
batch sizes and parallelizes additional tokens at near-zero marginal 
wall-clock cost.

Our model adopts the compute-bound formulation, constraining each batch 
to at most $B$ concurrent requests. This is also consistent with the fair serving constraint described above: when GPU compute drives batch cost, the per-step 
wall-clock time is determined by the \emph{longest} KV cache in the 
batch, creating the cross-subsidization of short requests by long ones. Extending the model to the 
memory-bound regime can be an interesting direction for future work.
\end{remark}}

\subsection{Examples of Unfair Schedules} \label{subsec:example}

\begin{example}
    \textbf{First-Come, First-Served (FCFS) } scheduling policy can lead to highly unfair and inefficient outcomes, as demonstrated in the Figure below. The scenario depicts three requests with demands $o_1 > o_2 \gg o_3$ arriving in sequence. Under FCFS, $r_3$ is batched with request 1's long tail. This forces $r_3$, which has a negligible resource requirement, to endure a latency comparable to $r_1$, severely inflating its completion time and leading to inefficient resource utilization.
    \begin{figure}[H]
  \centering
  \includegraphics[width=0.5\textwidth]{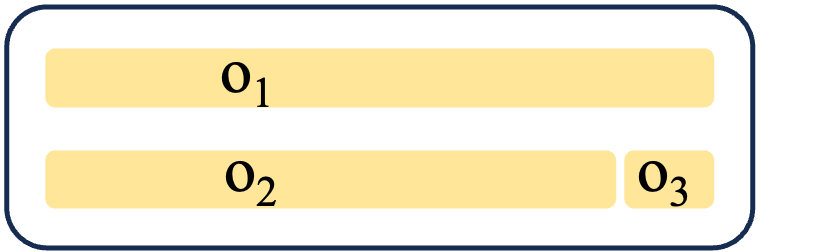}
  \label{fig:sub1}
\end{figure}
\end{example}

\begin{example}
    \textbf{Shortest First (SF) scheduling}, while often efficient for average latency, introduces a different form of unfairness that can lead to starvation (Figure below). Consider five requests arriving in sequence with demands $o_1 > o_2 > o_3 > o_4 > o_5$. For request 3, 2, and 1: each is forced to batch with a progressively larger tail to complete. As the executing batch aggregates more and larger jobs, its prolonged execution time exacerbates the unfairness for subsequent requests.
\begin{figure}[H]
  \centering
  \includegraphics[width=0.7\textwidth]{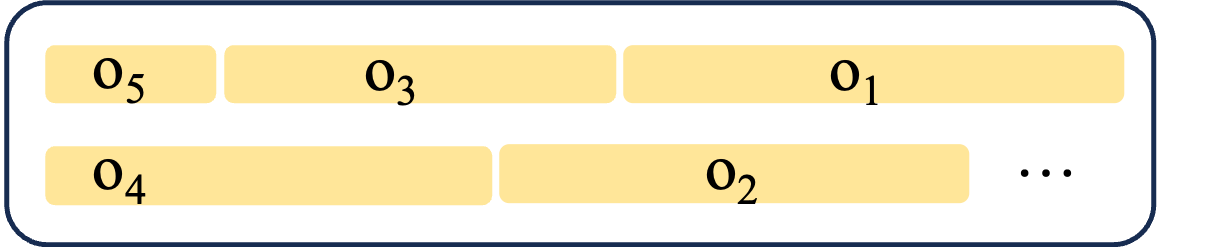}
  \label{fig:sub2}
\end{figure}
\end{example}

{
\begin{remark}[Continuous batching in FCFS and SJF]
\label{rem:cont-batch-examples}
Both FCFS and SJF operate under \emph{continuous batching}: whenever a request
completes and frees a processing slot, the next request (in arrival order
under FCFS, or by shortest remaining length under SJF) is immediately admitted
to the active batch without waiting for other running requests to finish.
The unfairness illustrated in the examples above therefore arises not from
static batch assignments, but from the \emph{head-of-line externality} that
persists even under continuous batching: as long as a long request remains
active, every short request admitted to share its batch is forced to advance
in lockstep with it, paying a per-step cost proportional to the long
request's KV footprint rather than its own.
\end{remark}
}

\section{Fair Scheduling Algorithms} \label{sec:alg}

This section presents two fair scheduling algorithms: Longest Job First (LJF) and Insert Short Jobs with Limit (ISJL). LJF provides absolute fairness, satisfying the fairness constraint for any $\alpha \geq 0$, but achieves poor throughput, serving as a benchmark for worst-case throughput. To address this limitation, we design ISJL, a parametric algorithm that accepts a fairness parameter $\alpha$ and provides a tunable trade-off, guaranteeing fairness for the given $\alpha$ while significantly improving throughput.

\subsection{Longest Job First (LJF)} \label{subsec:LJF}

First, we propose the Longest Job First (LJF) scheduling algorithm. LJF operates in discrete batches. It begins by selecting the $B$ longest requests from the queue to process simultaneously. The system then idles until the longest request in the current batch completes, upon which the entire batch is considered finished. Only then does LJF select the next $B$ longest requests to form a new batch. This process repeats indefinitely. Figure \ref{fig:ljf} provides an example of how LJF processes each batch of requests.

\begin{figure}[H]
  \centering
\includegraphics[width=0.8\textwidth]{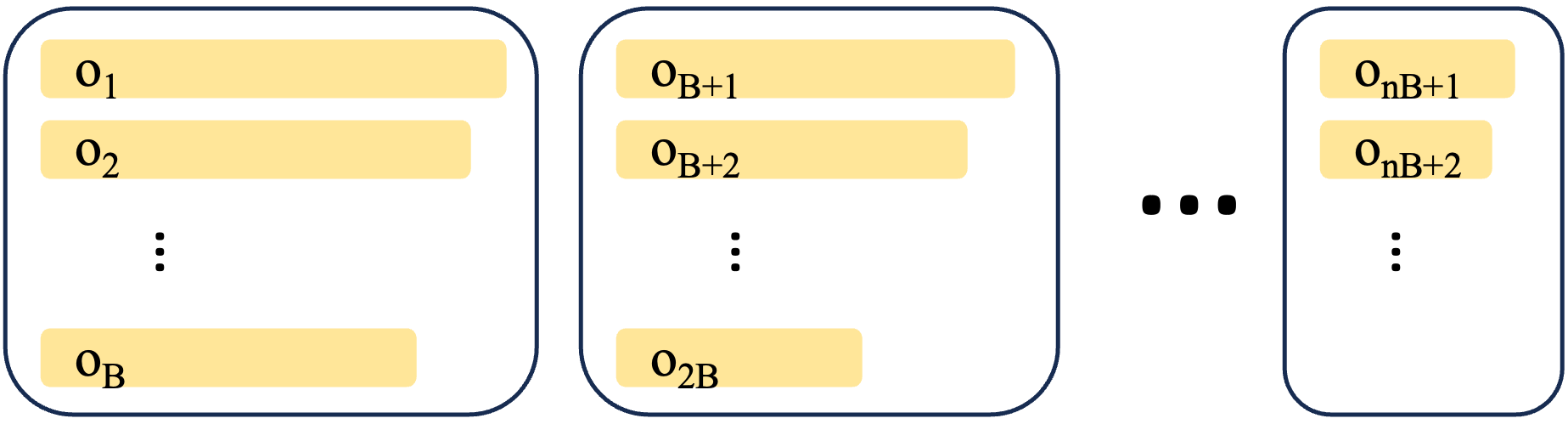}
\caption{An example of LJF}
  \label{fig:ljf}
\end{figure}

\begin{proposition} \label{prop:1}
    Algorithm LJF is fair, which satisfies the constraint \eqref{eq:fairness} for any $\alpha > 0$.
\end{proposition}

Proposition \ref{prop:1} states that LJF is fair for any $\alpha > 0$. The proof is straightforward: within any batch $S^{(t)}$, all requests start simultaneously. Therefore, the instantaneous acceleration for each request is equal, i.e., $\max_{i \in S^{(t)}}{a_i^{(t)}} = \min_{i \in S^{(t)}}{a_i^{(t)}}$, satisfying the fairness condition.
However, a significant drawback of LJF is its poor throughput. The algorithm does not admit new requests into a batch until the entire batch completes, even if resources are released early when shorter jobs finish. This leads to unavoidable resource idling. The following theorem quantifies this inefficiency.

\begin{algorithm}[H]
\caption{Longest Job First (LJF)}
\label{alg:ljf}
\begin{algorithmic}[1]
  \STATE $o = \{o_1 \ge o_2 \ge \dots \ge o_n\}$
  \FOR{$i = 1 \dots \lfloor n/B\rfloor$}
    \STATE \textbf{run} $\{o_{(i-1)B+1}, \dots , \,o_{iB}\}$
  \ENDFOR
  \STATE \textbf{run} the remaining jobs in another batch
\end{algorithmic}
\end{algorithm}

{ 
\begin{remark}[LJF as a static-batching baseline]
\label{rem:ljf-static}
Algorithm~\ref{alg:ljf} is the only \emph{static-batching} algorithm in this
paper: it forms a complete batch of the $B$ longest requests and holds it until
every member finishes, tolerating resource idling whenever shorter jobs complete
early, since no new request is admitted until the entire batch is done.
\end{remark}
}

\begin{theorem}[Competitive Ratio of Algorithm LJF]
\label{thm:ljf-competitive}
  For any $B \geq 1$, we have 
  \[
  \text{CR}(\text{LJF}) \ge \frac{B}{2B-1}.
  \]
\end{theorem}

\proof{Proof of Theorem~\ref{thm:ljf-competitive}.}

We begin with a trivial lemma states that maximizing throughput is equivalent to minimizing the makespan since we cannot cancel any request in our model.

\begin{lemma}[Makespan-Throughput Correspondance] \label{lem:1}
 Define the makespan of Algorithm $\mathcal A$ under the arrival instance $\mathcal I$ is $T_{\mathcal A}(\mathcal I)$. Then, maximizing throughput is equivalent to minimizing makespan. i.e.,
    \[
    \min \; T_{\mathcal A} (\mathcal{I}) \quad \Longleftrightarrow \quad \max \; \text{Throughput}(\mathcal{A}; \mathcal{I})
    \]
\end{lemma}

\proof{Proof of Lemma~\ref{lem:1}}
  For the given $n$ jobs in an arrival instance $\mathcal I$, $\text{Throughput}(\mathcal{A}; \mathcal{I})$ satisfies
  \[
    \text{Throughput}(\mathcal{A}; \mathcal{I}) = \frac{\sum_{i=1}^{n}o_i}{T_{\mathcal A} (\mathcal{I})},
  \]
  where $\sum_{i=1}^{n}o_i$ is unchanged within the same instance, which completes the proof.
\Halmos
\endproof

Then, we show that the competitive ratio of LJF is lower bounded by $\frac{B}{2B-1}$.

  W.L.O.G., let $nB$ jobs have processing times
  $o_1\ge o_2\ge\cdots\ge o_{nB}$, and Algorithm~1 produces the makespan
  \[
    T_{\mathcal LJF} (\mathcal{I}) \;=\; o_1 + o_{B+1} + o_{2B+1} + \dots + o_{(n-1)B+1}.
  \]
  Set $L:=o_1$, $O:=\sum_{i=1}^{nB}o_i$ and $T_{\mathcal OPT} (\mathcal{I})$ to be the makespan of optimal scheduling for instance $\mathcal{I}$.
  Any feasible schedule must process the longest job and must split the total work over batch size B, hence
  \[
    T_{\mathcal OPT} (\mathcal{I}) \;\ge\; L,
    \qquad
    T_{\mathcal OPT} (\mathcal{I}) \;\ge\; \frac{O}{B}.
  \]

  \paragraph{Case 1: $L \ge \frac{B}{2B-1} T_{\mathcal LJF} (\mathcal{I})$.}
  Then $T_{\mathcal LJF} (\mathcal{I}) \le \frac{2B-1}{B}\,L \le \frac{2B-1}{B}\,T_{\mathcal OPT} (\mathcal{I})$.

  \paragraph{Case 2: $L < \frac{B}{2B-1} T_{\mathcal LJF} (\mathcal{I})$.}
  Decompose
  \[
    T_{\mathcal LJF} (\mathcal{I}) \;=\; L + O_{\text{mod1}},
    \qquad  \text{where }
    O_{\text{mod1}} := o_{B+1}+o_{2B+1}+\dots+o_{(n-1)B+1}.
  \]
  The assumption yields that $O_{\text{mod1}} > \frac{B-1}{2B-1} T_{\mathcal LJF} (\mathcal{I})$.
  Because $o_{(k-1)B+2}\ge o_{(k-1)B+3} \ge \dots \ge o_{kB} \ge o_{kB+1}$ for all $k > 0$, therefore
  \[
    O_{\text{mod2}} := o_2+o_{B+2}+\dots+o_{(n-1)B+2}
    \;\ge\; O_{\text{mod1}}
    \;>\; \frac{B-1}{2B-1} T_{\mathcal LJF} (\mathcal{I}).
  \]
  \[
    O_{\text{mod3}} := o_3+o_{B+3}+\dots+o_{(n-1)B+3}
    \;\ge\; O_{\text{mod1}}
    \;>\; \frac{B-1}{2B-1} T_{\mathcal LJF} (\mathcal{I}).
  \]
  \[
    \vdots
  \]
  \[
    O_{\text{mod0}} := o_B+o_{2B}+\dots+o_{nB}
    \;\ge\; O_{\text{mod1}}
    \;>\; \frac{B-1}{2B-1} T_{\mathcal LJF} (\mathcal{I}).
  \]
  Consequently,
  \[
    O \;=\; T_{\mathcal LJF} (\mathcal{I}) + O_{\text{mod2}} + O_{\text{mod3}} + \dots + O_{\text{mod0}}
    \;>\; T_{\mathcal LJF} (\mathcal{I}) + \frac{(B-1)^2}{2B-1}\,T_{\mathcal LJF} (\mathcal{I})
        \;=\; \frac{B^2}{2B-1}T_{\mathcal LJF} (\mathcal{I}),
  \]
  \[\Longrightarrow\qquad
    T_{\mathcal OPT} (\mathcal{I}) \;\ge\; \frac{O}{B}
      \;>\; \frac{B}{2B-1}T_{\mathcal LJF} (\mathcal{I}).
  \]
  In both cases $T_{\mathcal LJF} (\mathcal{I}) \le \frac{2B-1}{B}\,T_{\mathcal OPT} (\mathcal{I})$, completing the proof.
\Halmos
\endproof

\begin{remark}
    For $B \geq 2$, the competitive ratio of LJF decreases monotonically as the batch size $B$ increases. The competitive ratio upper bound achieves its maximum value of $\frac{2}{3}$ when $B=2$.
\end{remark}

\subsection{Better Algorithm: ISJL (Warm-Up case, B=2)}

Motivated by the trade-off between the absolute fairness of LJF (fair for any $\alpha \geq 0$) and its severe throughput limitations, we propose a parametric algorithm that permits a bounded fairness violation (within a parameter $\alpha \geq 0$) to achieve significant throughput gains. This algorithm, called Insert Short Jobs with Limit (ISJL), takes $\alpha$ as its input parameter.

Given the complexity of ISJL, we first describe and analyze it for the warm-up case of $B=2$. The algorithm consists of two phases.

\textit{Phase One (Batch Initialization): } The algorithm first selects the longest waiting job, denoted $o_1$. To determine which jobs can be batched with $o_1$ without exceeding the allowed fairness violation, a threshold is calculated:
$q = \min\{o_1 - o_2 + \beta, 2\alpha\}$,
where $o_2$ is the second longest waiting job, and $\beta$ is a parameter related to $\alpha$, with $\beta \in [0,2\alpha]$. The algorithm then prioritizes selecting the longest jobs whose processing time is less than or equal to $q$ to run concurrently with $o_1$. This process continues until adding the next available job would cause the cumulative size of the batched short jobs to exceed the threshold $q$. At this point, the algorithm proceeds the second longest job, $o_2$, alongside $o_1$. Since Phase One does not guarantee that $o_1$ and $o_2$ will finish simultaneously, a Phase Two is required to schedule new jobs as resources become available. 

Exceptionally, when $o_2$ is less or equal to $\alpha$, which allows for $o_3$, the third longest waiting job to be processed right after $o_2$ and in parallel with $o_1$, Phase TWO is triggered directly.

\textit{Phase TWO (Batch Finalization): } This phase is triggered when a resource is freed, either by the completion of a job, or by the length of $o_2$. The algorithm checks the current decode length of the occupied resource, denoted by $a_m$. As long as $a_m \le \alpha$, the algorithm assign the longest unprocessed job to the freed resource.

The algorithm then iterates between Phase One and Phase Two, dynamically forming new batches or adding jobs to existing ones as resources become available, until all requests are completed. The complete formal procedure for ISJL is presented in Algorithm \ref{alg:isj-lim}. Next, Theorem \ref{thm:cr-isj-limited} shows that the lower bound of competitive ratio of ISJL is a function of $\beta$, with the best choice $\beta=\frac{\sqrt{17}-1}{4}\alpha \approx 0.78\alpha$ gives the maximum lower bound at $\frac{25+\sqrt{17}}{38} \approx 0.77$.

\begin{algorithm}
\caption{Insert Short Jobs with Limit (ISJL)}
\label{alg:isj-lim}
\begin{algorithmic}[1]
\STATE \textbf{Input: } $\alpha \geq 0$, $\beta = \frac{\sqrt{17}-1}{4}\alpha$.
  \STATE \textbf{Phase ONE}
  \STATE \textbf{Denote} $o = \{o_1 \ge o_2 \ge \dots \ge o_n\}$ for all undo jobs
  \IF {$o_2 \le \alpha$}
    \STATE \textbf{goto Phase TWO}
  \ENDIF
  \STATE $o_k \leftarrow \max\{o_i \in o \mid o_i \le \min\{o_1-o_2+\beta, 2\alpha\}\}$
  \STATE $o_{\text{array}} \leftarrow [\,]$
  \FOR{$i = k,\dots,n$}
    \IF{$o_i + \sum o_{\text{array}} > \min\{o_1-o_2+\beta, 2\alpha\}$}
      \STATE \textbf{continue}
    \ELSE
      \STATE \textbf{append} $o_i$ to $o_{\text{array}}$
    \ENDIF
  \ENDFOR
  \IF{$\sum o_{\text{array}} > \alpha$}
    \STATE \textbf{run} $o_{\text{array}}$ alone for $\sum o_{\text{array}} - \alpha$
  \ENDIF
  \STATE \textbf{run} $o_1$ with the rest part of $o_{\text{array}}$
  \STATE \textbf{run} $o_2$ with the rest part of $o_1$
  \STATE \textbf{Phase TWO}
  \WHILE {One of the resource is freed}
    \STATE $o_m \leftarrow \max\{o_i \in o \mid a_i = 0\}$
    \STATE $a_m \leftarrow $ the current output length on the occupied resource
    \IF {$a_m \le \alpha$}
      \STATE \textbf{run} $o_m$ until one of the resource is freed
    \ELSE
      \STATE \textbf{break}
    \ENDIF
  \ENDWHILE
  \STATE \textbf{repeat from step 1}
\end{algorithmic}
\end{algorithm}

{
\begin{remark}[ISJL as a hybrid continuous-batching policy]
\label{rem:isjl-hybrid}
ISJL occupies a deliberate middle ground between static
and continuous batching. Phase~ONE is semi-static: it constructs a structured short-job pack
around the longest active request $o_1$ and holds this configuration for a
bounded period, ensuring that all co-batched requests enter the decode
phase with KV caches within $\alpha$ tokens of each other.
\emph{Phase~TWO} is fully event-driven and implements continuous batching:
the moment a processing slot is freed and the current maximum progress
satisfies $A_{\max}(t)\le\alpha$, the next longest unprocessed request is
immediately admitted. This hybrid design is intentional.
The semi-static Phase~ONE enforces the fairness constraint by carefully
controlling the initial KV-cache spread of each new cohort of requests.
The continuous-batching Phase~TWO then recaptures efficiency by
eliminating idle time as soon as fairness permits new admissions.
\end{remark}
}

\begin{theorem}[Competitive ratio of Algorithm~\ref{alg:isj-lim}]
\label{thm:cr-isj-limited}
For $B = 2$, we have
\[
\text{CR}(\text{ISJL})\;\ge
\begin{cases}
\min\{\frac{\alpha+2\beta}{\alpha+3\beta},\,\frac{4\alpha}{6\alpha-\beta}\}, & \beta \in [0, \alpha] \\[10pt]
\min\{\frac{2\alpha+\beta}{2\alpha+2\beta},\,\frac{4\alpha}{6\alpha-\beta}\}, & \beta \in (\alpha, 2\alpha]
\end{cases}
\]
\end{theorem}

\proof{Proof of Theorem~\ref{thm:cr-isj-limited}.}
Let $T_{\mathcal ISJL} (\mathcal{I})$ be the makespan produced by Algorithm~\ref{alg:isj-lim} and let $T_{\mathcal OPT} (\mathcal{I})$ be the optimal makespan for the same input instance.  Then by Lemma \ref{lem:1}, it suffices to show that
\[
\inf_{\mathcal I}\frac{T_{\mathcal OPT}(\mathcal I)}{T_{\text{ISJL}}(\mathcal I)}\;=
\begin{cases}
\min\{\frac{\alpha+2\beta}{\alpha+3\beta},\,\frac{4\alpha}{6\alpha-\beta}\}, & \beta \in [0, \alpha] \\[10pt]
\min\{\frac{2\alpha+\beta}{2\alpha+2\beta},\,\frac{4\alpha}{6\alpha-\beta}\}, & \beta \in (\alpha, 2\alpha]
\end{cases}
\]

We first quantify the loss in Phase ONE, starting by define a term \textit{stage}: 
\begin{definition}
  \label{def:stage}
  A stage is initiated whenever line~2 of Algorithm~\ref{alg:isj-lim} is reached. For the $r$th stage, denote the current largest two undo jobs by $L_r$ (length $o_1$) and $S_r$ (length $o_2$), and the pack of short jobs by $P_r$. One \emph{stage} in Phase ONE is formally defined by the following sequence of actions:
  \begin{enumerate}[label=\roman*)]
  \item Execute $P_r$ alone if $P_r > \alpha$ for $P_r-\alpha$ tokens;
  \item Execute $L_r$ in parallel with $P_r$ until $P_r$ finishes;
  \item Execute $L_r$ in parallel with $S_r$ until either $L_r$ or $S_r$ finishes.
  \end{enumerate}
\end{definition}

Next, we point out that the loss in any stage of Phase ONE can produce loss from two aspects: over-inserting or under-inserting (See Figure \ref{fig:over-under} for intuition). Lemma \ref{lem:idle-time-lim} quantify the worst competitive ratio in Phase ONE under the over-inserting case, while Lemma \ref{lem:under-insert-lim} provides a property to bound the competitive ratios in the under-inserting case.

\begin{figure}[htbp]
  \centering
\includegraphics[width=0.8\textwidth]{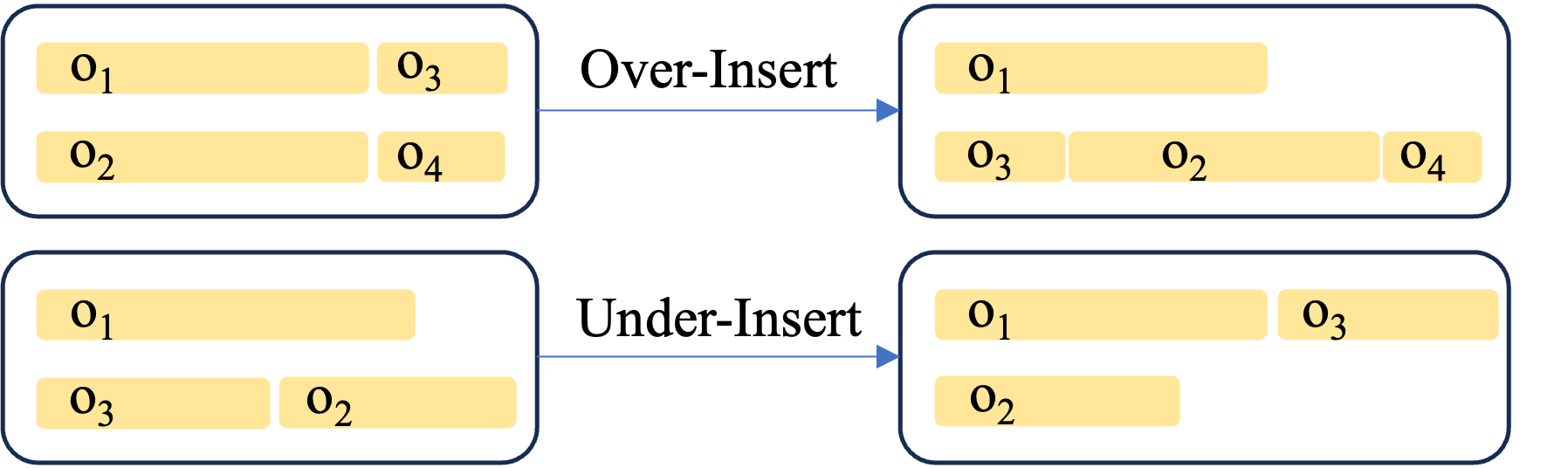}
  \caption{Over and Under Inserting}
  \label{fig:over-under}
\end{figure}

\begin{lemma}[Cost of over-inserting~\ref{alg:isj-lim}]
  \label{lem:idle-time-lim}
  For any stage $r$ of Phase ONE, the worst competitive ratio caused by over-inserting is given by:
  \[
  \text{CR}(\text{ISJL})\;\ge
  \begin{cases}
  \dfrac{\alpha+2\beta}{\alpha+3\beta}, & \beta \in [0, \alpha] \\[10pt]
  \dfrac{2\alpha+\beta}{2\alpha+2\beta}, & \beta \in (\alpha, 2\alpha]
  \end{cases}
  \]
\end{lemma}

\proof{Proof of Lemma~\ref{lem:idle-time-lim}.}
Denote $Q_r$ as the package of small jobs that could be executed in parallel with $P_r$ in optimal scheduling. Underlying constraints for over-inserting are
\[
L_r \ge S_r \ge P_r \ge Q_r \ge 0,\; P_r \le \min\{L_r -S_r + \beta, 2\alpha\},\; L_r > \min\{2\alpha,\,P_r+\alpha\},\; S_r+P_r>L_r.
\]
Because
\[
\text{CR}(\text{ISJL})=\frac{L_r+P_r}{S_r+P_r+Q_r}
\]
strictly increases in $L_r$ and strictly decreases in $S_r$ and $Q_r$, the minimum is approached at the boundary $S_r=L_r$ and $Q_r=P_r$. The constraints reduce to
\[
0<P_r\le \beta,\quad P_r\le2\alpha,\quad L_r>\min\{2\alpha,P_r+\alpha\}.
\]
For fixed $L_r$ the function $(L_r + P_r)/(L_r + 2P_r)$ decreases in $P_r$, so we take $P_r$ as large as allowed; for fixed $P_r$ it increases in $L_r$, so we take $L_r$ as small as allowed. Two cases:
\begin{itemize}
    \item If $0\le t\le\alpha$, then $P_r<t\le\alpha$, hence $\min\{2\alpha,P_r+\alpha\}=P_r+\alpha$ and the smallest admissible $L_r$ is $L_r\downarrow P_r+\alpha$. Letting $P_r\uparrow \beta$ gives
    \[
    \inf \text{CR}(\text{ISJL})=\frac{\alpha+2\beta}{\alpha+3\beta}.
    \]
    \item If $\alpha<\beta\le2\alpha$, we can choose $P_r$ with $\alpha\le P_r<\beta$, so $\min\{2\alpha,P_r+\alpha\}=2\alpha$ and the smallest admissible $L_r$ is $L_r\downarrow2\alpha$. Letting $P_r\uparrow \beta$ gives
    \[
    \inf \text{CR}(\text{ISJL})=\frac{2\alpha+\beta}{2\alpha+2\beta}=\frac{2\alpha+\beta}{2(\alpha+\beta)}\le \frac34,    
    \]
    matching continuously the first case at $\beta=\alpha$.
\end{itemize}
The overall infimum decreases from $1$ ($\beta=0$) down to $2/3$ ($\beta=2\alpha$), completing the proof.
\Halmos 
\endproof

\begin{lemma}[Cost of under-inserting]
  \label{lem:under-insert-lim}
  In any stage $r$ of Phase ONE, we have 
  \[
  \text{CR}(\text{ISJL})\ \ge\ \frac{4\alpha}{6\alpha-\beta}
  \]
\end{lemma}

\proof{Proof of Lemma~\ref{lem:under-insert-lim}.}
Underlying constraints for under-inserting are
\[
L_r \ge S_r \ge P_r \ge 0,\; L_r - S_r + \beta < P_r \le 2\alpha.
\]
From $L_r-S_r+\beta<P_r$ we have $L_r< S_r+P_r-\beta$.
For fixed $S_r,P_r$, $\text{CR}(\text{ISJL})=(S_r+P_r)/(L_r+P_r)$ is minimized by taking $L_r$ as large as allowed, i.e., $L_r\uparrow S_r+P_r-\beta$. Hence

$$
\text{CR}(\text{ISJL})\ \ge\ \frac{S_r+P_r}{\,S_r+2P_r-\beta\,}.
$$

For fixed $\beta,P_r$ with $P_r>\beta$, the right-hand side increases in $S_r$, so the minimum over $S_r\ge P_r$ occurs at $S_r=P_r$:

$$
\text{CR}(\text{ISJL})\ \ge\ \frac{2P_r}{\,3P_r-\beta\,}.
$$

For fixed $\beta$, this expression decreases in $P_r$, so using the bound $P_r\le 2\alpha$ and feasibility $P_r>\beta$ gives
\[
\text{CR}(\text{ISJL})\ \ge\ \frac{4\alpha}{\,6\alpha-\beta\,}.
\]

The bound is tight: it is approached by choices $S_r=P_r=2\alpha$ and $L_r$ arbitrarily close to $S_r+P_r-\beta$ from below.
When $\beta=0$, this yields $\text{CR}(\text{ISJL})\ge 2/3$.
\Halmos 
\endproof

Therefore, with Lemmas \ref{lem:idle-time-lim} and \ref{lem:under-insert-lim}, we obtain the adversarial competitive ratio in Phase ONE:
\[
\frac{T_{\mathcal OPT}(\mathcal I)}{T_{\text{ISJL}}(\mathcal I)}\;\ge
\begin{cases}
\min\{\dfrac{\alpha+2\beta}{\alpha+3\beta},\,\dfrac{4\alpha}{6\alpha-\beta}\}, & \beta \in [0, \alpha] \\[10pt]
\min\{\dfrac{2\alpha+\beta}{2\alpha+2\beta},\,\dfrac{4\alpha}{6\alpha-\beta}\}, & \beta \in (\alpha, 2\alpha]
\end{cases}
\]

Consider Phase TWO, we introduce the following proposition proved in \cite{graham1969bounds}:

\begin{proposition}[Theorem~2 in \cite{graham1969bounds}]
\label{prop:graham}
    Without the limit $\alpha$, the makespan of Phase TWO is at most $(\frac43 - \frac1{3B})\,T_{\mathcal OPT} (\mathcal{I})$.
\end{proposition}

Noticing that the competitive ratio in Phase TWO is larger than that obtained in Phase ONE, for $B=2$. Therefore, in both Phase ONE and Phase TWO we have
\[
\frac{T_{\mathcal OPT}(\mathcal I)}{T_{\text{ISJL}}(\mathcal I)}\;\ge
\begin{cases}
\min\{\dfrac{\alpha+2\beta}{\alpha+3\beta},\,\dfrac{4\alpha}{6\alpha-\beta}\}, & \beta \in [0, \alpha] \\[10pt]
\min\{\dfrac{2\alpha+\beta}{2\alpha+2\beta},\,\dfrac{4\alpha}{6\alpha-\beta}\}, & \beta \in (\alpha, 2\alpha]
\end{cases}
\]
which completes the proof.
\Halmos 
\endproof

\begin{remark}
    The best choice for parameter $\beta$ can be obtained by solving the inequality. Noticing the monotony of the two functions, the best $\beta$ solves
\[
\frac{\alpha+2\beta}{\alpha+3\beta}
=
\frac{4\alpha}{6\alpha-\beta}.
\]
Therefore, $\beta=\frac{\sqrt{17}-1}{4}\alpha \approx 0.78\alpha$ gives the maximum lower bound for the competitive ratio at $\frac{25+\sqrt{17}}{38} \approx 0.77$. 
\end{remark}

\begin{remark}
\label{rem:naive}
    A naive choice for $\beta$ is $\beta=\alpha$, which gives a competitive ratio of $\frac45$ in under-inserting and a competitive ratio of $\frac34$ in over-inserting. Therefore, the overall competitive ratio of this simplified choice is bounded by $\frac34=0.75$.
\end{remark}

\subsection{Generalization ($B \geq 2$)}

Motivated by the fact that the competitiveness of the insertion rule used in Phase TWO approaches $3/4$ as $B$ grows, as shown in Proposition~\ref{prop:graham}, and to keep the mechanism simple, we generalize ISJL by fixing $\beta=\alpha$. Let the unprocessed job lengths be ordered $o=\{o_1\ge o_2\ge \dots \ge o_n\}$. For $B\ge 2$ workers, define $M=\min\{B,\;|\{i:\,o_i>\alpha\}|\}$. If $o_2\le \alpha$, we immediately invoke Phase TWO; otherwise we build a batch around $o_1$ by allocating short jobs to accompany each of the next $M-1$ longest “long” jobs. For each $j\in\{2,\dots,M\}$ set a per-worker packing budget $\mathrm{cap}_j=\min\{2\alpha,\;o_1-o_j+\alpha\}$. Let $o_k$ be the largest job length not exceeding $\min\{2\alpha,\;o_1-o_M+\alpha\}$; we then scan the jobs $o_i$ for $i\ge k$ in decreasing order and greedily assign each to the largest index $j\in\{2,\dots,M\}$ whose current packed total plus $o_i$ stays within $\mathrm{cap}_j$. This produces disjoint short-job packs $\{w_{\text{array}}^j\}_{j=2}^M$ with $\sum w_{\text{array}}^j\le \mathrm{cap}_j$ by construction.

Scheduling proceeds in two stages while holding $o_1$ on worker $1$. First, on each worker $j\in\{2,\dots,M\}$ we run a prefix of its packed short jobs of total length exactly $\alpha$ in parallel; this ensures that no worker advances more than $\alpha$ units ahead of the slowest active job and thus bounds the fairness violation by $\alpha$. Second, we continue on each worker $j$ with the remaining tail of its pack, of length $\mathrm{tail}_j=\max\{0,\sum w_{\text{array}}^j-\alpha\}$, still in parallel with $o_1$. These two stages constitute Phase ONE for the current batch.

Phase TWO handles fill-in while any resources are idle. Let $A_{\max}$ be the maximum remaining processing time across all jobs currently running. As long as $A_{\max}\le \alpha$, we place the longest unprocessed job on any idle worker; if $A_{\max}>\alpha$, we stop admitting new jobs and return to Phase ONE to rebuild packs around the current longest jobs. The algorithm alternates between these two phases, rebuilding batches as the set of unprocessed jobs changes, until all work completes. With $\beta=\alpha$, the packing budgets $\mathrm{cap}_j$ and the $\alpha$-prefix rule together enforce the $\alpha$-bounded fairness constraint while exploiting the available parallelism across $B$ workers.

\begin{algorithm}
\caption{Generalized Insert Short Jobs with Limit (ISJL, $B\ge 2$)}
\label{alg:isj-lim-gen}
\label{alg:isjl-b}
\begin{algorithmic}[1]
  \STATE \textbf{Input: } $\alpha \ge 0$, number of workers $B\ge 2$
  \STATE \textbf{Phase ONE}
  \STATE \textbf{Denote} $o = \{o_1 \ge o_2 \ge \dots \ge o_n\}$ for all undo jobs
  \STATE $M \leftarrow \min\bigl\{B,\; |\{i\mid o_i>\alpha\}|\bigr\}$, $o_k \leftarrow \max\{o_i \mid o_i \le \min\{2\alpha,\, o_1 - o_M + \alpha\}\}$
  \FOR{$j = 2,\dots,B$}
      \STATE $w_{\text{array}}^j \leftarrow [\,]$
  \ENDFOR
  \FOR{$j = 2,\dots,M$}
      \STATE $\text{cap}_j \leftarrow \min\{2\alpha,\, o_1 - o_j + \alpha\}$
  \ENDFOR
  \FOR{$i = k,\dots,n$}
  \STATE $\mathrm{assigned}\leftarrow \mathrm{false}$
  \FOR{$j = M+1,\dots,B$}
      \IF{$\mathrm{assigned}=\mathrm{false}$ and $o_i + \sum w_{\text{array}}^j \le \alpha$}
          \STATE \textbf{append} $o_i$ to $w_{\text{array}}^j$; $\mathrm{assigned}\leftarrow \mathrm{true}$
      \ENDIF
  \ENDFOR
  \FOR{$j = M,\dots,2$}
      \IF{$\mathrm{assigned}=\mathrm{false}$ and $o_i + \sum w_{\text{array}}^j \le \mathrm{cap}_j$}
          \STATE \textbf{append} $o_i$ to $w_{\text{array}}^j$; $\mathrm{assigned}\leftarrow \mathrm{true}$
      \ENDIF
  \ENDFOR
\ENDFOR
  \FOR{$j = 2,\dots,B$}
      \STATE \textbf{run} a prefix of $w_{\text{array}}^j$ of length $\alpha$ on worker $j$ in parallel
      \STATE \textbf{idle} worker $1$ for $\min\{\max\{\sum o_{\text{array}}^M-\alpha,0\},\alpha\}$
  \ENDFOR
  \FOR{$j = 2,\dots,B$}
      \STATE \textbf{run} the remaining part of $w_{\text{array}}^j$ of length $\text{tail}_j$ on worker $j$, in parallel with $o_1$
  \ENDFOR

  \STATE \textbf{Phase TWO}
  \WHILE{a resource is idle}
      \STATE $A_{\max}\leftarrow$ current maximum output length among all running jobs
      \STATE $o_m \leftarrow \max\{o_i \in o \mid o_i \text{ is undo}\}$ 
      \IF {$A_{\max} \le \alpha$}
          \STATE \textbf{run} $o_m$ on an arbitrary idle worker
      \ELSE
          \STATE \textbf{break}
      \ENDIF
  \ENDWHILE
  \STATE \textbf{repeat from Phase ONE with the remaining undo jobs}
\end{algorithmic}
\end{algorithm}

\begin{theorem}[Competitive ratio of Algorithm~\ref{alg:isj-lim-gen}]
\label{thm:cr-isj-lim-gen}
For any $B \ge 2$, we have
\[
\text{CR}(\text{ISJL}) \ge \frac34. 
\]
\end{theorem}

\proof{Proof of Theorem~\ref{thm:cr-isj-lim-gen}.}
By Proposition~\ref{prop:graham}, we know that the competitive ratio of Phase TWO is bounded by $\frac{3B}{4B-1} > \frac34$. Consider two cases in Phase ONE:
\begin{enumerate}[label=\roman*)]
\item (Under-inserting) Inserting algorithm is not optimal in Phase ONE, i.e., $\lvert P_r \rvert < \lvert P_{opt} \rvert$.
\item (Over-inserting) Inserting P causes additional idle time, i.e., another job that could have been run in parallel is now running after $L$ and $S$ (which requires that $L_r < S_r + \lvert P_r \rvert$).
\end{enumerate}

In case i), for any $i \in [2,B]$, consider $o_1$ and $o_i$, idle time is bounded by Lemma~\ref{lem:under-insert-lim}.

In case ii), for any $i \in [2,B]$, consider $o_1$ and $o_i$, idle time is bounded by Lemma~\ref{lem:idle-time-lim}.

Therefore, the lower bound of overall competitive ratio is given by the minimum over the two phases,
\[
\text{CR}(\text{ISJL}) \ge \frac34, 
\]
completing the proof.
\Halmos 
\endproof

\subsection{Robust Scheduling under Prediction Uncertainty}
\label{sec:robust}

In practice, the exact decoding length $o_i$ of request $i$ is unknown upon arrival. Instead, the system
may have access to a prediction interval $[\ell_i,u_i]$ such that $\ell_i \le o_i \le u_i$.
We define the (minimum) prediction-interval accuracy as
\[
\kappa := \min_i \frac{\ell_i}{u_i}\in(0,1].
\]

\paragraph{Robust-ISJL (Upper-Bound ISJL).}
We define \textbf{Robust-ISJL} by running Algorithm~\ref{alg:isj-lim-gen} 
using the proxy sizes $\widehat{o}_i := u_i$ wherever job sizes are used for sorting, forming packs,
and computing quantities such as $M$ and $cap_j$.
During execution, the fairness check always uses the \emph{true run-time progress}
$\{a_i(t)\}$ (e.g., the Phase TWO condition $A_{\max}(t)\le \alpha$ is evaluated using the current
maximum progress among the active jobs, as in Algorithm~\ref{alg:isj-lim-gen}). The following theorem shows that our Robust-ISJL satisfies the fairness constraint, and the proof can be found in Appendix \ref{append:alg}.

\begin{theorem}[Fairness of Robust-ISJL]
\label{thm:robust-fairness}
Robust-ISJL satisfies the fairness constraint~\eqref{eq:fairness} for any $\alpha \ge 0$.
\end{theorem}

\paragraph{Why a naive $\kappa$-multiplicative bound is invalid.}
A tempting approach is to prove $T_{\text{OPT}}(I_u)\le \rho T_{\text{OPT}}(I)$ for $\rho:=1/\kappa$ by time-stretching
an optimal schedule for $I$. This is false under the token-gap fairness constraint~\eqref{eq:fairness}:
time-stretching can amplify progress gaps and violate feasibility (it would generally require relaxing
$\alpha$ to $\rho\alpha$).

\paragraph{A prediction-error dependent guarantee.}
Define the request-level upper-bound slack
\[
e_i:=u_i-o_i\ge 0,
\]
and the aggregate upper-bound overestimation error
\[
\Delta:=\sum_{i=1}^n e_i=\sum_{i=1}^n (u_i-o_i).
\]

Let \(\mathcal P_{\alpha,B}(\mathcal I)\) denote the family of partitions of \([n]\) into groups \(G\) such that \(|G|\le B\) and
\[
\max_{i\in G} o_i-\min_{i\in G} o_i\le \alpha .
\]
For each request, define the upper-bound slack \(e_i:=u_i-o_i\). The packable residual prediction error is
\[
\Delta_{\alpha,B}^{\mathrm{pack}}(\mathcal I,u)
:=
\min_{\mathcal P\in\mathcal P_{\alpha,B}(\mathcal I)}
\sum_{G\in\mathcal P}\max_{i\in G} e_i .
\]

\begin{theorem}[Competitive guarantee of Robust-ISJL under interval uncertainty]
\label{thm:robust-cr}
For any $B\ge 2$ and any instance with $\ell_i\le o_i\le u_i$,
Robust-ISJL satisfies
\[
T_{\mathrm{Robust\text{-}ISJL}}(\mathcal I)
\le
\frac{4}{3}
\left(
T_{\mathrm{OPT}}(\mathcal I)
+
\Delta_{\alpha,B}^{\mathrm{pack}}(\mathcal I,u)
\right).
\]
Equivalently,
\[
\frac{T_{\mathrm{OPT}}(\mathcal I)}
{T_{\mathrm{Robust\text{-}ISJL}}(\mathcal I)}
\ge
\frac{3}{4}\cdot
\frac{1}
{1+
\Delta_{\alpha,B}^{\mathrm{pack}}(\mathcal I,u)/
T_{\mathrm{OPT}}(\mathcal I)}.
\]
In particular, when predictions are perfect, i.e., $u_i=o_i$ for all $i$,
we have $\Delta_{\alpha,B}^{\mathrm{pack}}(\mathcal I,u)=0$ and recover
$\mathrm{CR}(\mathrm{Robust\text{-}ISJL})\ge 3/4$.
\end{theorem}

\begin{proof}{Proof of Theorem~\ref{thm:robust-cr}.}
Let $\mathcal I_u$ denote the proxy instance whose request lengths are exactly
$u_i$. Running Algorithm~\ref{alg:isj-lim-gen} on $\mathcal I_u$
yields makespan $T_{\mathrm{ISJL}}(\mathcal I_u)$. Executing the same batching
decisions on the realized instance $\mathcal I$ cannot take longer because each
request requires only $o_i\le u_i$ service; early completions only remove jobs
from the active set. Hence,
\[
T_{\mathrm{Robust\text{-}ISJL}}(\mathcal I)
\le
T_{\mathrm{ISJL}}(\mathcal I_u).
\]

By Theorem~\ref{thm:cr-isj-lim-gen},
\[
T_{\mathrm{ISJL}}(\mathcal I_u)
\le
\frac{4}{3}
T_{\mathrm{OPT}}(\mathcal I_u).
\]

We upper-bound $T_{\mathrm{OPT}}(\mathcal I_u)$ using an
optimal $\alpha$-fair schedule for the realized instance $\mathcal I$. After
this schedule completes, request $i$ has received $o_i$ units of service. To
complete the proxy instance $\mathcal I_u$, request $i$ requires an additional
\[
e_i=u_i-o_i
\]
units of service.

Fix any partition
$\mathcal P\in\mathcal P_{\alpha,B}(\mathcal I)$. For each group
$G\in\mathcal P$, process all requests in $G$ together for
$\max_{i\in G} e_i$ time steps, removing request $i$ once it has received its
additional $e_i$ units of service. Capacity feasibility follows from
$|G|\le B$.

We now verify fairness. At the beginning of the residual processing for group
$G$, the progress values of requests in $G$ are $\{o_i:i\in G\}$. By the
definition of $\mathcal P_{\alpha,B}(\mathcal I)$,
\[
\max_{i\in G} o_i-\min_{i\in G} o_i\le \alpha.
\]
While active requests in $G$ are processed together, their progress values
increase in lockstep, so the max--min progress gap among active requests is
unchanged. When a request completes its residual work and is removed, the
max--min gap cannot increase. Therefore, each residual group is processed by a
feasible $\alpha$-fair schedule.

The residual tail length induced by partition $\mathcal P$ is
\[
\sum_{G\in\mathcal P}\max_{i\in G} e_i.
\]
Minimizing over all feasible partitions yields
\[
T_{\mathrm{OPT}}(\mathcal I_u)
\le
T_{\mathrm{OPT}}(\mathcal I)
+
\Delta_{\alpha,B}^{\mathrm{pack}}(\mathcal I,u).
\]
Combining (1)--(3) gives the stated result.
\Halmos
\end{proof}

\begin{corollary}[Scale-normalized robustness]
\label{cor:robust-scale-normalized}
Let
\[
O(\mathcal I):=\sum_{i=1}^n o_i,
\qquad
\bar\varepsilon(\mathcal I):=
\frac{\Delta}{O(\mathcal I)}
=
\frac{\sum_{i=1}^n (u_i-o_i)}{\sum_{i=1}^n o_i}
\]
denote the total realized token workload and the average upper-bound slack per
realized token. Then Robust-ISJL satisfies
\[
\frac{T_{\mathrm{OPT}}(\mathcal I)}
{T_{\mathrm{Robust\text{-}ISJL}}(\mathcal I)}
\ge
\frac{3}{4}\cdot
\frac{1}{1+B\bar\varepsilon(\mathcal I)}.
\]
In particular, if $u_i\le (1+\varepsilon)o_i$ for all $i$, then
\[
\frac{T_{\mathrm{OPT}}(\mathcal I)}
{T_{\mathrm{Robust\text{-}ISJL}}(\mathcal I)}
\ge
\frac{3}{4}\cdot
\frac{1}{1+B\varepsilon}.
\]
\end{corollary}

The proof of Corollary~\ref{cor:robust-scale-normalized} can be found in Appendix~\ref{append:alg}.

\begin{remark}[Large-scale interpretation of the robust guarantee]
\label{rem:robust-large-scale}
Theorem~\ref{thm:robust-cr} strengthens the aggregate-error bound by recognizing
that residual overestimation errors can often be processed in fair batches. The
old aggregate term $\Delta=\sum_i(u_i-o_i)$ corresponds to the pessimistic
singleton partition, in which every residual error is processed separately.
When many requests have realized lengths within the same $\alpha$-window, the
effective residual penalty can be much smaller. For example, if all realized
lengths lie within one $\alpha$-window and all residual errors equal $e$, then
\[
\Delta=ne,
\qquad
\Delta_{\alpha,B}^{\mathrm{pack}}
=
\left\lceil\frac{n}{B}\right\rceil e,
\]
so the aggregate-error bound can overstate the residual penalty by a factor
approaching $B$.

Corollary~\ref{cor:robust-scale-normalized} clarifies the large-scale behavior.
Although $\Delta$ grows with the number of requests, the optimal makespan also
grows with the realized token workload. For fixed batch size $B$, the guarantee
depends on the average upper-bound slack $\bar\varepsilon(\mathcal I)=\Delta/O(\mathcal I)$,
not directly on $n$. Thus, large-scale operation does not by itself weaken the
robustness certificate. The certificate deteriorates when prediction intervals
are systematically loose relative to realized output lengths. In deployment,
this means Robust-ISJL should use calibrated upper bounds with controlled
average slack, rather than extremely conservative global upper bounds such as
the maximum context length for every request. Fairness remains robust because
admission decisions are always checked using realized run-time progress; interval
errors affect only the planning order and hence the efficiency guarantee.
\end{remark}

\subsection{Online Arrivals: An Extension and Its Limitation} \label{sec:online-discussion}

Consider an online release-date instance
\[
\mathcal I^r=\{(r_i,o_i)\}_{i=1}^n,
\]
where request $i$ is revealed to the scheduler at release time $r_i$. At time
$t$, the scheduler knows only the released unfinished requests. Let $Q(t)$ denote
the queue of released requests that have not yet started service.

Online-ISJL implements Algorithm~\ref{alg:isj-lim-gen} in a rolling manner. When
Phase ONE is invoked at time $t$, the algorithm sorts only the currently
available queue $Q(t)$ and computes $M$, $o_k$, and the packing budgets $cap_j$
using this released set. Requests that arrive during execution are inserted into
$Q(t)$ but do not change the currently running batch. During Phase TWO, whenever
a slot becomes free, the algorithm admits the longest currently released request
in $Q(t)$ if the usual fairness check $A_{\max}(t)\le \alpha$ is satisfied; if
the check fails, the algorithm returns to Phase ONE. If $Q(t)$ is empty, the
server waits until the next request arrives.

\begin{proposition}[Fairness of Online-ISJL]
\label{prop:online-isjl-fair}
For any arrival sequence $\{(r_i,o_i)\}_{i=1}^n$ and any $\alpha\ge 0$,
Online-ISJL satisfies the fairness constraint~\eqref{eq:fairness} at every time
step.
\end{proposition}

Proposition \ref{prop:online-isjl-fair} states that our Online-ISJL is fair and the proof can be found in Appendix \ref{append:alg}. Although it is hard to provide an upper bound of the competitive ratio of Online-ISJL, the next proposition provides a lower bound to the online setting. 

\begin{proposition}
\label{prop:online-lower-bound}
Fix any finite $\alpha\ge 0$ and set $B=2$. In the online release-date model
under the same nonpreemptive service convention as the main algorithms, no
deterministic online scheduling policy can guarantee competitive ratio $3/4$
against a clairvoyant offline optimum that knows all release times in advance.
More precisely, for every deterministic online policy and every $\delta>0$,
there exists an online instance such that
\[
\frac{T_{\mathrm{OPT}}^r(\mathcal I^r)}
     {T_{\mathcal A}^r(\mathcal I^r)}
\le
\frac{\sqrt 5-1}{2}+\delta
<
\frac34,
\]
where $T_{\mathrm{OPT}}^r$ is the optimal $\alpha$-fair makespan respecting
release times.
\end{proposition}

\begin{proof}{Proof of Proposition~\ref{prop:online-lower-bound}.}
Fix a large length $L$ and release one request $J_1$ of length $L$ at time
$0$. Let $s$ be the time at which the deterministic online policy first starts
processing $J_1$. If the policy never starts $J_1$, the adversary releases no
more requests and the competitive ratio is zero. Hence assume $s<\infty$.

The adversary considers two continuations.

First, suppose no additional request arrives. Then the online policy completes
no earlier than $s+L$, whereas the offline optimum starts $J_1$ at time $0$ and
completes at time $L$. Thus,
\[
\frac{T_{\mathrm{OPT}}^r}{T_{\mathcal A}^r}
\le
\frac{L}{s+L}.
\]

Second, suppose a second request $J_2$ of length $L$ is released at time
\[
r=s+\alpha+\varepsilon,
\]
where $\varepsilon>0$. At time $r$, request $J_1$ has already advanced by
$\alpha+\varepsilon$ tokens, whereas $J_2$ has progress zero. Therefore $J_2$
cannot be admitted into the same active batch as $J_1$ without violating the
$\alpha$-fairness constraint. Under the nonpreemptive service convention,
the online policy must complete $J_1$ before processing $J_2$, so
\[
T_{\mathcal A}^r\ge s+2L.
\]
The clairvoyant offline optimum instead waits until time $r$ and then starts
$J_1$ and $J_2$ together. Their progress values remain equal throughout service,
so the schedule is $\alpha$-fair and completes at time
\[
T_{\mathrm{OPT}}^r\le r+L=s+\alpha+\varepsilon+L.
\]
Thus,
\[
\frac{T_{\mathrm{OPT}}^r}{T_{\mathcal A}^r}
\le
\frac{s+\alpha+\varepsilon+L}{s+2L}.
\]

Let $x=s/L$ and $\eta=(\alpha+\varepsilon)/L$. Combining the two possible
continuations, the adversary can force
\[
\frac{T_{\mathrm{OPT}}^r}{T_{\mathcal A}^r}
\le
\min\left\{
\frac{1}{1+x},
\frac{1+x+\eta}{2+x}
\right\}.
\]
Taking $L$ sufficiently large makes $\eta$ arbitrarily small. Hence the best
ratio any deterministic online policy can guarantee is at most
\[
\sup_{x\ge 0}
\min\left\{
\frac{1}{1+x},
\frac{1+x}{2+x}
\right\}.
\]
The first function is decreasing in $x$, and the second is increasing in $x$.
The maximum of the minimum occurs when they are equal:
\[
\frac{1}{1+x}
=
\frac{1+x}{2+x}
\quad\Longleftrightarrow\quad
x^2+x-1=0.
\]
Thus $x=(\sqrt 5-1)/2$, and the common value is also $\frac{\sqrt 5-1}{2}$.
\Halmos
\end{proof}

Proposition~\ref{prop:online-lower-bound} clarifies the scope of the offline
competitive-ratio result. 
Online-ISJL remains event-driven and causal. The difficulty is the release-time
information structure. A clairvoyant offline scheduler can delay an available
long request to wait for a future compatible request, whereas an online scheduler
does not know whether such a request will arrive. This waiting-versus-processing
trade-off is absent from the offline backlog model and requires a different
online analysis.

Accordingly, the $3/4$-competitive theorem should be read as an offline guarantee for a backlogged request pool. In deployment, Online-ISJL
is a natural rolling implementation: it applies the same resource-fair packing
logic to the currently released queue and admits new requests whenever the
runtime fairness check permits. Obtaining positive competitive guarantees under
stochastic arrivals, bounded lookahead, or batching-window policies is an
important direction for future work.

\section{Theoretical Extensions: Tighter Competitive Ratio as a Function of $\alpha$} \label{sec:extension}

In Section \ref{sec:alg}, we established a lower bound of \(\frac{3}{4}\) on the competitive ratio of Algorithm ISJL, which holds for any value of \(\alpha\) (Theorem \ref{thm:cr-isj-lim-gen}). This result prompts two natural questions: How does the competitive ratio vary with \(\alpha\)? And for which values of \(\alpha\) is the \(\frac{3}{4}\) bound tight? To address these questions, we derive a tighter, parameterized characterization of the competitive ratio. Our approach involves categorizing adversarial instances into three distinct types and then determining which type dictates the worst-case performance for a given \(\alpha\). For consistency with Algorithm~\ref{alg:isj-lim-gen}, we conduct this analysis with the parameter \( \beta = \alpha\).


\subsection{Adversarial Behaviors}

To analyze the competitive ratio with respect to every $\alpha$, we start with discussing two adversarial behaviors: \textit{over-inserting} and \textit{under-inserting}. Following a similar line of reasoning as in the previous proof, we begin by examining the case $B = 2$, and then generalize the argument to the case $B \geq 2$.

\paragraph{Over-Inserting}

As illustrated in Fig.~\ref{fig:over-under}, the over-inserting behavior is characterized by the constraints $o_3 \le o_1 - o_2 + \alpha$, $o_3 \le 2\alpha$ and $o_1 > \min\{o_3+\alpha,2\alpha\}$. The competitive ratio is then given by
\[
\mathrm{CR}_{\mathrm{over}}=\dfrac{o_1+o_3}{o_2+o_3+o_4}. 
\]

\begin{lemma}[Over-Inserting: $B=2$]
\label{lem:gamma-over}
For each fixed $\gamma\in(0,1)$, the greatest lower bound of $\mathrm{CR}_{\mathrm{over}}$ over the feasible region satisfies
\[
\mathrm{CR}_{\mathrm{over}}(\gamma)\ge
\begin{cases}
\dfrac{1+\gamma}{1+2\gamma}, & 0<\gamma\le \tfrac12,\\[8pt]
\dfrac{2-\gamma}{3-2\gamma}, & \tfrac12\le \gamma<1.
\end{cases}
\]
\end{lemma}

The proof can be found in Appendix \ref{append:extend}.

\medskip
\paragraph{Under-Inserting}

Under-inserting yields an instance in which placing $o_3$ before $o_2$ leads to a shorter makespan. To construct such a case, it suffices to require $o_1 - o_2 + \alpha < o_3 \le 2\alpha$. Consider $\mathrm{CR}_{\mathrm{under}}=\dfrac{o_2+o_3}{o_1+o_3}$. 

\begin{lemma}[Under-Inserting: $B=2$]
\label{lem:gamma-under}
\[
\mathrm{CR}_{\textnormal{under}}(\gamma)\ge
\begin{cases}
\dfrac{1+\gamma}{1+2\gamma}, & 0<\gamma\le \tfrac{1}{3},\\[6pt]
\dfrac{2(1+\gamma)}{3+\gamma}, & \tfrac{1}{3}\le \gamma<1,
\end{cases}
\]
$\mathrm{CR}_{\textnormal{under}}(\gamma)$ is strictly decreasing on $(0,\tfrac{1}{3}]$ and strictly increasing on $[\tfrac{1}{3},1)$, with a unique global minimum at $\gamma=\tfrac{1}{3}$ equal to $\tfrac{4}{5}$; finally, $\lim_{\gamma\to 0^+}\mathrm{CR}_{\textnormal{under}}(\gamma)=\lim_{\gamma\to 1^-}\mathrm{CR}_{\textnormal{under}}(\gamma)=1$.
\end{lemma}

The proof can be found in Appendix \ref{append:extend}.

\medskip

\paragraph{Jigsaw Effect}

\begin{figure}[htbp]
\centering
\includegraphics[width=0.8\textwidth]{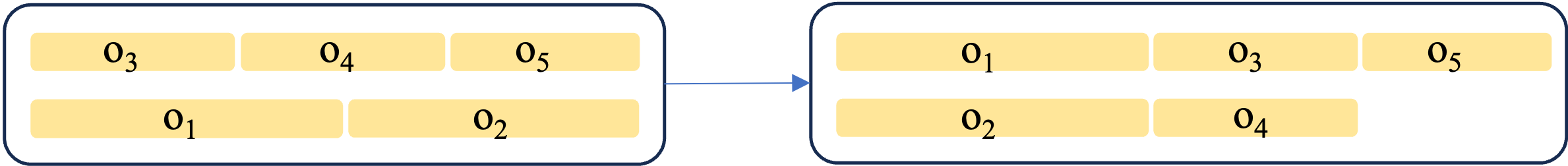}
\caption{Jigsaw Effect}
\label{fig:jigsaw}
\end{figure}

As shown in Fig.~\ref{fig:jigsaw}, the jigsaw effect arises when concatenating a set of smaller jobs that happens to produce a makespan comparable to that of another set of larger jobs. The adversarial case can be formulated as 
\[
o=\{2B-1,2B-1,2B-2,2B-2,…,B+1,B+1,B,B,B\}.
\]
Recall that Proposition~\ref{prop:graham} gives a competitive ratio no less than $\frac{3B}{4B-1}$ under this adversary, and thus the bound won't be effective because $\frac{3B}{4B-1} > \frac34$, indicating that the jigsaw effect is dominated by over-inserting.

\subsection{Dominant Behavior Respecting $\gamma$ and Generalization}

Let $\gamma=\alpha/o_1\in(0,1)$ and denote by $\underline{\mathrm{CR}}(\gamma)$ the greatest lower bound of $\mathrm{CR}(\mathrm{ISJL})$ over all feasible instances with fixed~$\gamma$. From the preceding analysis in Lemma~\ref{lem:gamma-over} and Lemma~\ref{lem:gamma-under} we have the behavior-specific bounds
\[
\mathrm{CR}_{\mathrm{over}}(\gamma)=
\begin{cases}
\dfrac{1+\gamma}{1+2\gamma}, & 0<\gamma\le \tfrac12,\\[6pt]
\dfrac{2-\gamma}{3-2\gamma}, & \tfrac12\le \gamma<1,
\end{cases}
\qquad
\mathrm{CR}_{\mathrm{under}}(\gamma)=
\begin{cases}
\dfrac{1+\gamma}{1+2\gamma}, & 0<\gamma\le \tfrac13,\\[6pt]
\dfrac{2(1+\gamma)}{3+\gamma}, & \tfrac13\le \gamma<1.
\end{cases}
\]
The overall worst-case lower bound is the lower envelope
\[
\underline{\mathrm{CR}}(\gamma)=\min\{\mathrm{CR}_{\mathrm{over}}(\gamma),\ \mathrm{CR}_{\mathrm{under}}(\gamma)\}.
\]
The following theorem explicitly gives the piecewise function of $\mathrm{CR}_{\mathrm{under}}(\gamma)$ and proves that this lower bound is the same for $B=2$ and $B>2$ cases.

\begin{theorem}[Generalized global competitive ratio as a function of $\gamma$]
\label{thm:gamma}

$\forall B \ge 2$, The overall worst-case lower bound $\underline{\mathrm{CR}}(\gamma)$ admits the single piecewise description
\[
\underline{\mathrm{CR}}(\gamma)=
\begin{cases}
\dfrac{1+\gamma}{1+2\gamma}, & 0<\gamma\le \tfrac12,\\[8pt]
\dfrac{2-\gamma}{3-2\gamma}, & \tfrac12\le \gamma<1,
\end{cases}
\]
which is strictly decreasing on $(0,\tfrac12]$, strictly increasing on $[\tfrac12,1)$, attains its unique global minimum $\underline{\mathrm{CR}}(\tfrac12)=\tfrac34$, and satisfies
\[
\lim_{\gamma\to0^+}\underline{\mathrm{CR}}(\gamma)=\lim_{\gamma\to1^-}\underline{\mathrm{CR}}(\gamma)=1.
\]
\end{theorem}

\proof{Proof of Theorem~\ref{thm:gamma}}
When $B=2$, it suffices to compare the behavior-specific bounds on the relevant sub-intervals and to establish monotonicity on each branch. In this proof, we start by discussing three regimes of $\gamma$ under $B=2$. And in the last part of the proof, we generalize this result to $B \ge 2$ cases.

\paragraph{Small $\gamma$ ($0<\gamma\le \tfrac{1}{3}$).}
In this range the two adversarial patterns coincide:
\[
\mathrm{CR}_{\mathrm{over}}(\gamma)=\mathrm{CR}_{\mathrm{under}}(\gamma)=\frac{1+\gamma}{1+2\gamma}.
\]
Hence
\[
\underline{\mathrm{CR}}(\gamma)=\frac{1+\gamma}{1+2\gamma},
\]
which is strictly decreasing on $(0,\tfrac13]$, with $\lim_{\gamma\to0^+}\underline{\mathrm{CR}}(\gamma)=1$ and $\underline{\mathrm{CR}}(\tfrac13)=\tfrac{4}{5}$.

\paragraph{Intermediate $\gamma$ ($\tfrac{1}{3}<\gamma\le \tfrac{1}{2}$).}
Here we compare the two candidate bounds:
\[
\frac{1+\gamma}{1+2\gamma}\ \le\ \frac{2(1+\gamma)}{3+\gamma}
\quad\Longleftrightarrow\quad
\frac{1}{1+2\gamma}\ \le\ \frac{2}{3+\gamma}
\quad\Longleftrightarrow\quad
\gamma\ge \tfrac13.
\]
Thus on $\gamma\in(\tfrac13,\tfrac12]$ the over-inserting bound dominates, and
\[
\underline{\mathrm{CR}}(\gamma)=\frac{1+\gamma}{1+2\gamma},
\]
which remains strictly decreasing up to $\underline{\mathrm{CR}}(\tfrac12)=\tfrac34$ by the derivative computed above.

\paragraph{Large $\gamma$ ($\tfrac{1}{2}\le\gamma<1$).}
For $\gamma\ge\tfrac12$, the over-inserting bound switches to
\[
\mathrm{CR}_{\mathrm{over}}(\gamma)=\frac{2-\gamma}{3-2\gamma},
\]
which is strictly increasing on $[\tfrac12,1)$. Moreover,
\[
\frac{2-\gamma}{3-2\gamma}\ \le\ \frac{2(1+\gamma)}{3+\gamma}
\quad\Longleftrightarrow\quad
(2-\gamma)(3+\gamma)\ \le\ 2(1+\gamma)(3-2\gamma)
\quad\Longleftrightarrow\quad
0\le 3\gamma(1-\gamma),
\]
which holds for all $\gamma\in[0,1]$. Hence the over-inserting bound is again the lower envelope on $[\tfrac12,1)$:
\[
\underline{\mathrm{CR}}(\gamma)=\frac{2-\gamma}{3-2\gamma}.
\]
Finally, $\lim_{\gamma\to1^-}\underline{\mathrm{CR}}(\gamma)=\lim_{\gamma\to1^-}\frac{2-\gamma}{3-2\gamma}=1$.

\paragraph{Generalization.}
$\forall i \in [2,B]$, consider $o_1$ and $o_i$, the total idle time is bounded by Lemma~\ref{lem:gamma-over} and Lemma~\ref{lem:gamma-under}. Therefore, the competitive ratio in Phase ONE of Algorithm~\ref{alg:isj-lim-gen} is bounded by $\underline{\mathrm{CR}}(\gamma)$. Proposition~\ref{prop:graham} ensures that Phase TWO won't make the competitive ratio worse. Combining the three regimes yields
\[
\underline{\mathrm{CR}}(\gamma)=
\begin{cases}
\dfrac{1+\gamma}{1+2\gamma}, & 0<\gamma\le \tfrac12,\\[8pt]
\dfrac{2-\gamma}{3-2\gamma}, & \tfrac12\le \gamma<1,
\end{cases}
\]
which is strictly decreasing on $(0,\tfrac12]$ and strictly increasing on $[\tfrac12,1)$. Since the two branches meet continuously at $\gamma=\tfrac12$ and are strictly monotone on either side, $\underline{\mathrm{CR}}$ attains a unique global minimum at $\gamma=\tfrac12$ with value $\tfrac34$, and satisfies the boundary limits $\lim_{\gamma\to0^+}\underline{\mathrm{CR}}(\gamma)=\lim_{\gamma\to1^-}\underline{\mathrm{CR}}(\gamma)=1$.
\Halmos 
\endproof

\begin{figure}[htbp]
\centering
\includegraphics[width=0.7\textwidth]{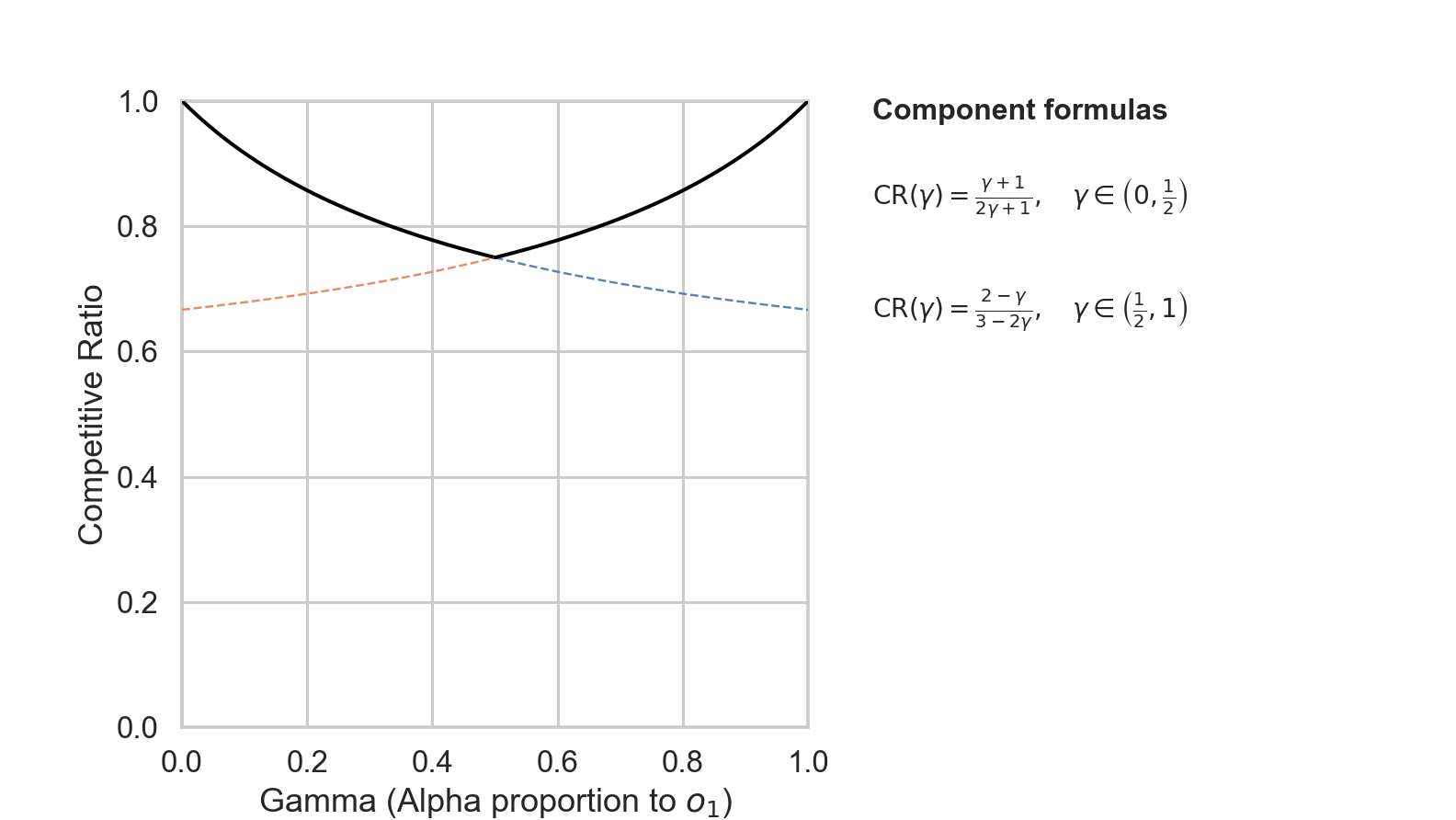}
\caption{Competitive Ratio as a Function of $\gamma$}
\label{fig:cr_alpha}
\end{figure}

\section{Profit and Cost Alignment under Linear Pricing}
\label{sec:profit-cost-alignment}

The preceding adversarial analysis establishes that ISJL has a worst-case
throughput guarantee under the resource-fairness constraint. We now connect this
throughput analysis to the operator's economic objective. Commercial LLM inference is typically sold through token-metered API contracts. Leading providers (OpenAI, Anthropic, Deepseek, Google Gemini) publish model-specific unit prices for tokens, usually quoted per one million tokens \citep{openai_api_pricing,anthropic_api_pricing,deepseek_api_pricing,google_gemini_api_pricing}. Thus, although providers may use different coefficients for tokens, the usage charge is linear in the realized token counts. Consistent with the token-level abstraction in Section~\ref{sec:model}, we represent this practice by a single effective linear price $p_i=\rho_0 o_i$ for request $i$ for some unit token price $\rho_0>0$. The economic tension studied below is that this token-metered revenue is additive across requests, whereas the cost of a batched inference step is driven by the largest active request in the batch. Resource-fair scheduling is therefore a mechanism for aligning linear token-metered revenue with max-driven batch cost.

\xhdr{Accepted workload and linear revenue.}
Fix an accepted request instance $\mathcal I=\{o_1,\ldots,o_n\}$ and let
\[
O(\mathcal I):=\sum_{i=1}^n o_i
\]
denote the total token workload. Under linear token pricing, the total revenue is
\begin{equation}
\label{eq:linear-revenue-fixed-workload}
R(\mathcal I)=p_0\sum_{i=1}^n o_i=p_0 O(\mathcal I),
\end{equation}
which is independent of the schedule used to serve the accepted workload.
Therefore, for a fixed accepted instance, maximizing profit is equivalent to
minimizing the schedule-dependent inference cost.

\xhdr{Token-progress cost.}
At the beginning of service step $t$, let $S^{(t)}$ be the active batch selected
by a schedule $\mathcal A$. Recall from Section~\ref{sec:model} that
$a_i^{(t)}$ is the number of tokens already processed for request $i$. Define $r_i^{(t)}:=a_i^{(t)}+1$, so that $r_i^{(t)}$ is the index of the token currently being processed. This
shift by one is only notational; it leaves the fairness gap unchanged, since $
\max_{i\in S^{(t)}} r_i^{(t)}-\min_{i\in S^{(t)}} r_i^{(t)}
=
\max_{i\in S^{(t)}} a_i^{(t)}-\min_{i\in S^{(t)}} a_i^{(t)}$.
Thus any schedule satisfying constraint~\eqref{eq:fairness} also satisfies $\max_{i\in S^{(t)}} r_i^{(t)}-\min_{i\in S^{(t)}} r_i^{(t)}\le \alpha$.

We define the instantaneous cost of a nonempty batch as
\begin{equation}
\label{eq:batch-cost-progress}
C_{\mathcal A}^{(t)}
=
 c+\tau |S^{(t)}|\max_{i\in S^{(t)}} r_i^{(t)},
\end{equation}
where $c\ge 0$ is the fixed cost of one service step and $\tau>0$ converts the
max-driven resource footprint into cost. The total inference cost and profit of
schedule $\mathcal A$ are then
\begin{equation}
\label{eq:total-cost-profit}
C_{\mathcal A}(\mathcal I)=\sum_{t=1}^{T_{\mathcal A}(\mathcal I)} C_{\mathcal A}^{(t)},
\qquad
\Pi_{\mathcal A}(\mathcal I)=R(\mathcal I)-C_{\mathcal A}(\mathcal I).
\end{equation}

\begin{theorem}[Profit decomposition and externality bound]
\label{thm:profit-decomposition}
For any feasible schedule $\mathcal A$ serving instance $\mathcal I$, define
\begin{equation}
\label{eq:intrinsic-resource-cost}
Q(\mathcal I):=\sum_{i=1}^n\sum_{k=1}^{o_i} k
=\frac12\sum_{i=1}^n o_i(o_i+1),
\end{equation}
and define the batching externality of $\mathcal A$ as
\begin{equation}
\label{eq:batching-externality}
E_{\mathcal A}(\mathcal I)
:=
\sum_{t=1}^{T_{\mathcal A}(\mathcal I)}
\sum_{i\in S^{(t)}}
\left(
\max_{j\in S^{(t)}} r_j^{(t)}-r_i^{(t)}
\right).
\end{equation}
Then profit admits the decomposition
\begin{equation}
\label{eq:profit-decomposition}
\Pi_{\mathcal A}(\mathcal I)
=
R(\mathcal I)
-
\tau Q(\mathcal I)
-
cT_{\mathcal A}(\mathcal I)
-
\tau E_{\mathcal A}(\mathcal I).
\end{equation}
Moreover, if $\mathcal A$ satisfies the fairness constraint~\eqref{eq:fairness}, then
\begin{equation}
\label{eq:externality-alpha-bound}
E_{\mathcal A}(\mathcal I)\le \alpha O(\mathcal I).
\end{equation}
\end{theorem}

The proof of Theorem~\ref{thm:profit-decomposition} can be found in Appendix \ref{append:profit}.
Theorem~\ref{thm:profit-decomposition} separates profit into four terms. The
revenue term $R(\mathcal I)$ and the intrinsic resource term $\tau Q(\mathcal I)$
are fixed once the accepted workload is fixed. The schedule affects profit only
through two terms: the time cost $cT_{\mathcal A}(\mathcal I)$ and the batching
externality cost $\tau E_{\mathcal A}(\mathcal I)$. The throughput analysis developed earlier controls the first schedule-dependent
term, while the fairness constraint controls the second. Thus the throughput theory is not
separate from the economic objective; it controls one component of profit, and
resource fairness controls the other.

\begin{theorem}[Additive profit guarantee for ISJL]
\label{thm:isjl-profit-guarantee}
Let $\mathcal F_\alpha(\mathcal I)$ be the set of schedules satisfying the
fairness constraint~\eqref{eq:fairness} on instance $\mathcal I$. Define
\[
\Pi_\alpha^\star(\mathcal I)
:=
\max_{\mathcal A\in\mathcal F_\alpha(\mathcal I)}
\Pi_{\mathcal A}(\mathcal I),
\]
and let $\mathcal A_\pi^\star\in\mathcal F_\alpha(\mathcal I)$ be a
profit-maximizing fair schedule with makespan
$T_\pi^\star(\mathcal I):=T_{\mathcal A_\pi^\star}(\mathcal I)$. Suppose a fair
algorithm $\mathcal A$ has throughput competitive ratio at least $\rho\in(0,1]$,
or equivalently,
\begin{equation}
\label{eq:rho-makespan-bound}
T_{\mathcal A}(\mathcal I)
\le
\frac{1}{\rho}T_{\mathcal OPT}(\mathcal I),
\end{equation}
where $T_{\mathcal OPT}(\mathcal I)=\min_{\mathcal B\in\mathcal F_\alpha(\mathcal I)}
T_{\mathcal B}(\mathcal I)$. Then
\begin{equation}
\label{eq:generic-profit-guarantee}
\Pi_{\mathcal A}(\mathcal I)
\ge
\Pi_\alpha^\star(\mathcal I)
-
c\left(\frac{1}{\rho}-1\right)T_\pi^\star(\mathcal I)
-
\tau\alpha O(\mathcal I).
\end{equation}
In particular, applying Theorem~\ref{thm:cr-isj-lim-gen} to ISJL gives
\begin{equation}
\label{eq:isjl-profit-guarantee}
\Pi_{\mathcal ISJL}(\mathcal I)
\ge
\Pi_\alpha^\star(\mathcal I)
-
\frac{c}{3}T_\pi^\star(\mathcal I)
-
\tau\alpha O(\mathcal I).
\end{equation}
\end{theorem}

The proof of Theorem~\ref{thm:isjl-profit-guarantee} can be found in Appendix \ref{append:profit}.
A multiplicative competitive ratio for profit is generally not the right object,
because profit may be close to zero or even negative depending on the price and
cost parameters. Theorem~\ref{thm:isjl-profit-guarantee} instead gives an
additive guarantee. ISJL can lose profit relative to the best fair profit
schedule only through two channels: an overhead term caused by imperfect
makespan optimality, and an externality term bounded by $\tau\alpha O(\mathcal I)$.
When $\alpha$ is small, ISJL is close to the cost alignment of LJF; when $c$ is
large, ISJL's throughput guarantee becomes economically valuable.

\begin{proposition}[Pairwise profit comparison]
\label{prop:pairwise-profit-comparison}
For any two feasible schedules $\mathcal A$ and $\mathcal B$ serving the same
accepted instance $\mathcal I$,
\begin{equation}
\label{eq:pairwise-profit-identity}
\Pi_{\mathcal A}(\mathcal I)-\Pi_{\mathcal B}(\mathcal I)
=
 c\left(T_{\mathcal B}(\mathcal I)-T_{\mathcal A}(\mathcal I)\right)
+
\tau\left(E_{\mathcal B}(\mathcal I)-E_{\mathcal A}(\mathcal I)\right).
\end{equation}
Consequently, since LJF has zero batching externality,
\begin{equation}
\label{eq:isjl-ljf-profit-comparison}
\Pi_{\mathcal ISJL}(\mathcal I)-\Pi_{\mathcal LJF}(\mathcal I)
\ge
c\left(T_{\mathcal LJF}(\mathcal I)-T_{\mathcal ISJL}(\mathcal I)\right)
-
\tau\alpha O(\mathcal I).
\end{equation}
Similarly, for FCFS,
\begin{equation}
\label{eq:isjl-fcfs-profit-comparison}
\Pi_{\mathcal ISJL}(\mathcal I)-\Pi_{\mathcal FCFS}(\mathcal I)
\ge
c\left(T_{\mathcal FCFS}(\mathcal I)-T_{\mathcal ISJL}(\mathcal I)\right)
+
\tau\left(E_{\mathcal FCFS}(\mathcal I)-\alpha O(\mathcal I)\right).
\end{equation}
\end{proposition}

The proof of Proposition~\ref{prop:pairwise-profit-comparison} can be found in Appendix \ref{append:profit}.
Proposition~\ref{prop:pairwise-profit-comparison} makes explicit the tradeoff
among the three main policies. LJF is the cost-alignment benchmark because it has
$E_{\mathcal LJF}=0$, but it may have a large makespan. FCFS may have small
makespan, but it can have a large externality term. ISJL lies between them: it
allows limited heterogeneity to recover throughput, while bounding the resulting
externality by $\alpha O(\mathcal I)$. In particular, ISJL is guaranteed to be
more profitable than LJF whenever
\[
c\left(T_{\mathcal LJF}(\mathcal I)-T_{\mathcal ISJL}(\mathcal I)\right)
\ge
\tau\alpha O(\mathcal I),
\]
and it is guaranteed to be more profitable than FCFS whenever the right-hand
side of~\eqref{eq:isjl-fcfs-profit-comparison} is nonnegative.

\begin{proposition}[FCFS can have quadratic batching externality]
\label{prop:fcfs-quadratic-externality}
Consider an instance with batch size $B=2$ consisting of one long request of
length $L$ and $L$ short requests of length $1$. The FCFS order places the long
request first, followed by the $L$ short requests. Then
\begin{equation}
\label{eq:fcfs-quadratic-externality}
T_{\mathcal FCFS}(\mathcal I)=L,
\qquad
E_{\mathcal FCFS}(\mathcal I)=\frac{L(L-1)}{2}.
\end{equation}
In contrast, LJF has zero batching externality and makespan
\begin{equation}
\label{eq:ljf-adversarial-makespan}
T_{\mathcal LJF}(\mathcal I)=L+\left\lceil\frac{L-1}{2}\right\rceil,
\qquad
E_{\mathcal LJF}(\mathcal I)=0.
\end{equation}
Therefore
\begin{equation}
\label{eq:ljf-fcfs-profit-adversarial}
\Pi_{\mathcal LJF}(\mathcal I)-\Pi_{\mathcal FCFS}(\mathcal I)
=
\tau\frac{L(L-1)}{2}
-
c\left\lceil\frac{L-1}{2}\right\rceil.
\end{equation}
In particular, for any fixed $c\ge 0$ and $\tau>0$, LJF is more profitable than
FCFS on this instance for all sufficiently large $L$.
\end{proposition}

The proof of Proposition~\ref{prop:fcfs-quadratic-externality} can be found in Appendix \ref{append:profit}.
Proposition~\ref{prop:fcfs-quadratic-externality} shows why a pure throughput
criterion can be economically misleading. FCFS finishes the adversarial instance
faster than LJF, but it creates a quadratic batching externality by repeatedly
co-batching a fresh short request with the increasingly advanced long request.
When the fixed per-step cost $c$ is small relative to the max-driven cost
coefficient $\tau$, the externality term dominates the additional steps required
by LJF. This explains why LJF can be the lowest-cost policy in numerical
experiments with very small $c$, while still being unattractive as a production
scheduler because of its inferior throughput and latency.

\begin{corollary}[Makespan minimization as certified profit maximization]
\label{cor:makespan-certified-profit}
Fix an accepted workload $\mathcal I$ and a fairness parameter $\alpha$.
For any $\alpha$-fair schedule $\mathcal A$,
\[
\Pi_{\mathcal A}(\mathcal I)
\ge
R(\mathcal I)-\tau Q(\mathcal I)-\tau \alpha O(\mathcal I)
-cT_{\mathcal A}(\mathcal I).
\]
Consequently, among schedules satisfying the same $\alpha$-fairness constraint,
minimizing makespan maximizes this certified lower bound on profit.
\end{corollary}

Corollary~\ref{cor:makespan-certified-profit} clarifies the role of the
throughput objective in Section~\ref{sec:model}. The proof can be found in Appendix \ref{append:profit}. We do not claim that the
operator only cares about makespan, nor that ISJL is the pure cost minimizer for
a fixed workload. Instead, $\alpha$ specifies the operator's acceptable
resource-misalignment budget, and the scheduling problem optimizes efficiency
within that budget. LJF corresponds to the strict endpoint $\alpha=0$ and is
therefore a natural cost-alignment benchmark. ISJL targets the interior of the
frontier: it permits a controlled amount of resource heterogeneity in order to
recover the throughput and latency benefits of continuous batching.

\subsection{Practical Calibration of the Fairness Parameter $\alpha$}
\label{subsec:alpha-calibration}

The parameter $\alpha$ controls how much resource
heterogeneity the scheduler is allowed to pool within the same active batch.
A smaller $\alpha$ enforces stricter homogeneity and reduces the batching
externality, but may restrict continuous batching and lower throughput. A
larger $\alpha$ increases batching flexibility, but can increase the
schedule-induced externality. Thus, $\alpha$ is a tuning parameter for the
operator's cost--service trade-off.

The theoretical results above provide two complementary calibration
principles. First, the feasible set expands with $\alpha$, so a larger
$\alpha$ weakly relaxes the scheduling constraint. Let $\mathcal F_\alpha$
denote the set of schedules satisfying the $\alpha$-fairness constraint, and
let
\[
T_\alpha^\star(\mathcal I)
=
\min_{\mathcal A\in\mathcal F_\alpha} T_{\mathcal A}(\mathcal I),
\qquad
\Theta_\alpha^\star(\mathcal I)
=
\frac{O(\mathcal I)}{T_\alpha^\star(\mathcal I)}
\]
denote the optimal makespan and optimal throughput under fairness budget
$\alpha$, where $O(\mathcal I)=\sum_i o_i$ is the total token workload.

\begin{proposition}[Nested throughput frontier]
\label{prop:nested-alpha-frontier}
If $0\le \alpha_1\le \alpha_2$, then
\[
\mathcal F_{\alpha_1}\subseteq \mathcal F_{\alpha_2},
\qquad
T_{\alpha_2}^\star(\mathcal I)\le T_{\alpha_1}^\star(\mathcal I),
\qquad
\Theta_{\alpha_2}^\star(\mathcal I)\ge \Theta_{\alpha_1}^\star(\mathcal I).
\]
\end{proposition}

The proof can be found in Appendix \ref{append:profit}. Second, the externality bound in Theorem~\ref{thm:profit-decomposition}
gives a direct cost-based way to choose an upper range for $\alpha$. Suppose
the operator wants the batching-externality cost to be no more than a fraction
$\delta$ of the intrinsic workload cost:
\[
\tau E_{\mathcal A}(\mathcal I)
\le
\delta \tau Q(\mathcal I).
\]
Since every $\alpha$-fair schedule satisfies
$E_{\mathcal A}(\mathcal I)\le \alpha O(\mathcal I)$, it is sufficient to
choose
\[
\alpha
\le
\alpha_{\mathrm{cost}}(\delta)
:=
\delta \frac{Q(\mathcal I)}{O(\mathcal I)}.
\]
Under the token-level abstraction with $r_i^{(t)}=a_i^{(t)}+1$,
\[
Q(\mathcal I)
=
\sum_{i=1}^n \sum_{k=1}^{o_i} k
=
\sum_{i=1}^n \frac{o_i(o_i+1)}{2},
\qquad
O(\mathcal I)=\sum_{i=1}^n o_i.
\]
Thus,
\[
\frac{Q(\mathcal I)}{O(\mathcal I)}
=
\frac{\sum_i o_i(o_i+1)}{2\sum_i o_i},
\]
which can be estimated directly from historical request logs. In large
samples, if $L$ denotes the random request length, this ratio is approximately
\[
\frac{\mathbb E[L^2]+\mathbb E[L]}{2\mathbb E[L]}.
\]
Therefore, the natural scale of $\alpha$ depends on the workload length
distribution; the same numerical value of $\alpha$ need not have the same
operational meaning across workloads.

In practice, we recommend the following calibration procedure. The operator
first estimates the request-length distribution from recent logs and constructs
a candidate grid $\mathcal G$ of $\alpha$ values, scaled by either
$Q(\mathcal I)/O(\mathcal I)$ or by empirical length quantiles. The operator
then replays historical workloads under ISJL for each $\alpha\in\mathcal G$
and records throughput, latency, and the cost decomposition
\[
C_{\mathrm{ISJL}}(\mathcal I;\alpha)
=
\tau Q(\mathcal I)
+
cT_{\mathrm{ISJL}}(\mathcal I;\alpha)
+
\tau E_{\mathrm{ISJL}}(\mathcal I;\alpha).
\]
Because $Q(\mathcal I)$ and revenue are fixed for the accepted workload, the
profit-relevant tuning problem is
\[
\alpha^\star
\in
\arg\min_{\alpha\in\mathcal G}
\left\{
cT_{\mathrm{ISJL}}(\mathcal I;\alpha)
+
\tau E_{\mathrm{ISJL}}(\mathcal I;\alpha)
\right\},
\]
possibly subject to service-level constraints such as
\[
\mathrm{Throughput}_{\mathrm{ISJL}}(\alpha)\ge \Theta^{\mathrm{SLA}},
\qquad
\mathrm{Latency}_{q,\mathrm{ISJL}}(\alpha)\le L_q^{\mathrm{SLA}},
\]
where $L_q^{\mathrm{SLA}}$ is a target quantile of end-to-end latency. Equivalently,
when a throughput or latency SLA is binding, the operator may choose the
smallest $\alpha$ that satisfies the service target, because smaller values
provide tighter cost alignment.

\subsection{Numerical Evidence}
\label{sec:profit-cost-numerics}

The preceding results show that the economic effect of scheduling can be
separated into a makespan component and a batching-externality component. We now
quantify these two components in simulation. The numerical design differs from
our earlier pricing-centered experiment in one important respect: the main text
uses only linear token pricing, and the accepted workload is fixed before
schedulers are compared. Thus, within each replication, FCFS, LJF, and ISJL face
exactly the same set of accepted requests and therefore generate exactly the
same token-metered revenue. Any difference in profit is consequently a
cost-side effect of scheduling.

\xhdr{Experimental design.}
We generate $n=400$ potential requests. Each request length is drawn
independently as
\[
        o_i\sim \mathrm{Uniform}\{1,2,\ldots,1000\}.
\]
Consistent with the linear-pricing model in
\eqref{eq:linear-revenue-fixed-workload}, request $i$ is quoted price
$p_i=p_0o_i$, where $p_0=p/500$ and $p=1$. We retain the same length-dependent
acceptance model used in the original numerical study. Specifically, requests
are divided into ten length buckets $[1,100],[101,200],\ldots,[901,1000]$, and a
request in bucket $b\in\{1,\ldots,10\}$ accepts with probability
\[
        \Pr(\text{accept})
        =\max\{0,1-k_b p_i\},
        \qquad
        k_b=\frac{1}{p_0(2b-1)100}.
\]
After acceptance is realized, the accepted workload is held fixed and served by
all scheduling policies. We compare FCFS, LJF, and ISJL with
$\alpha\in\{50,100,150,250, 300, 400, 500\}$ under batch size $B=50$. The cost parameters
are $c=0.0005$ and $\tau=10^{-6}$. For each schedule, we compute cost using
\eqref{eq:batch-cost-progress} and report the decomposition
\[
        C_{\mathcal A}(\mathcal I)
        =\tau Q(\mathcal I)
        +cT_{\mathcal A}(\mathcal I)
        +\tau E_{\mathcal A}(\mathcal I).
\]
The reported values are averages over five independent replications.

\begin{table}[t]
\centering
\caption{Profit and cost decomposition under linear pricing, \(n=400\), \(B=50\)}
\label{tab:cost-decomposition-n400-b50}
\begin{tabular}{lrrrrrrrr}
\toprule
Policy & Profit & Cost & \(\tau Q\) & \(cT\) & \(\tau E\) & Throughput & Gain vs. LJF & Avg. extent \\
\midrule
FCFS & 128.73 & 68.41 & 32.07 & 1.24 & 35.10 & 39.78 & -0.12\% & 566.26 \\
LJF & 163.82 & 33.31 & 32.07 & 1.24 & 0.00 & 39.83 & 0.00\% & 0.00 \\
ISJL(\(\alpha=50\)) & 162.53 & 34.60 & 32.07 & 1.20 & 1.33 & 41.23 & 3.50\% & 36.62 \\
ISJL(\(\alpha=100\)) & 161.10 & 36.03 & 32.07 & 1.16 & 2.79 & 42.30 & 6.21\% & 74.83 \\
ISJL(\(\alpha=150\)) & 159.59 & 37.55 & 32.07 & 1.12 & 4.35 & 44.00 & 10.48\% & 111.86 \\
ISJL(\(\alpha=250\)) & 156.30 & 40.83 & 32.07 & 1.05 & 7.71 & 46.99 & 17.97\% & 187.29 \\
ISJL(\(\alpha=300\)) & 154.71 & 42.43 & 32.07 & 1.03 & 9.32 & 47.92 & 20.31\% & 222.57 \\
ISJL(\(\alpha=400\)) & 151.31 & 45.83 & 32.07 & 1.02 & 12.74 & 48.36 & 21.42\% & 290.71 \\
ISJL(\(\alpha=500\)) & 147.82 & 49.32 & 32.07 & 1.01 & 16.23 & 48.61 & 22.03\% & 361.91 \\
\bottomrule
\end{tabular}
\end{table}

Table~\ref{tab:cost-decomposition-n400-b50} shows the three-way tradeoff among
FCFS, LJF, and ISJL. First, the intrinsic workload term \(\tau Q\) is identical
across all policies, as predicted by Theorem~\ref{thm:profit-decomposition}. Thus
the profit differences come only from the schedule-dependent terms \(cT\) and
\(\tau E\). Second, FCFS has essentially the same throughput as LJF in this
configuration, but it incurs a large batching-externality cost. Its average
within-batch extent is 566.26, and its externality cost is 35.10. This explains
why FCFS has substantially lower profit than the fair schedules under the same
linear token-pricing rule. Third, LJF eliminates the batching externality, but it
achieves this by static batching. It is therefore cost-aligned but less flexible.

ISJL traces a managerial frontier between these two extremes. For small
\(\alpha\), ISJL is close to LJF in cost alignment: at \(\alpha=50\), its profit is
within 0.8\% of LJF while increasing throughput by 3.5\%. As \(\alpha\) increases,
ISJL admits more jobs through continuous batching and obtains larger throughput
improvements at the cost of a controlled increase in batching externality. At
\(\alpha=300\), ISJL improves throughput by 20.31\% relative to LJF and still
improves profit by 20.19\% relative to FCFS; its profit remains within 5.56\% of
LJF. Larger values of \(\alpha\) continue to improve throughput slightly, but the
additional batching externality erodes profit. This pattern is consistent with
the theory: the fairness constraint bounds the externality term, while relaxing
\(\alpha\) increases scheduling flexibility and reduces the makespan term.

\begin{figure}[htbp]
\centering
\includegraphics[width=0.48\textwidth]{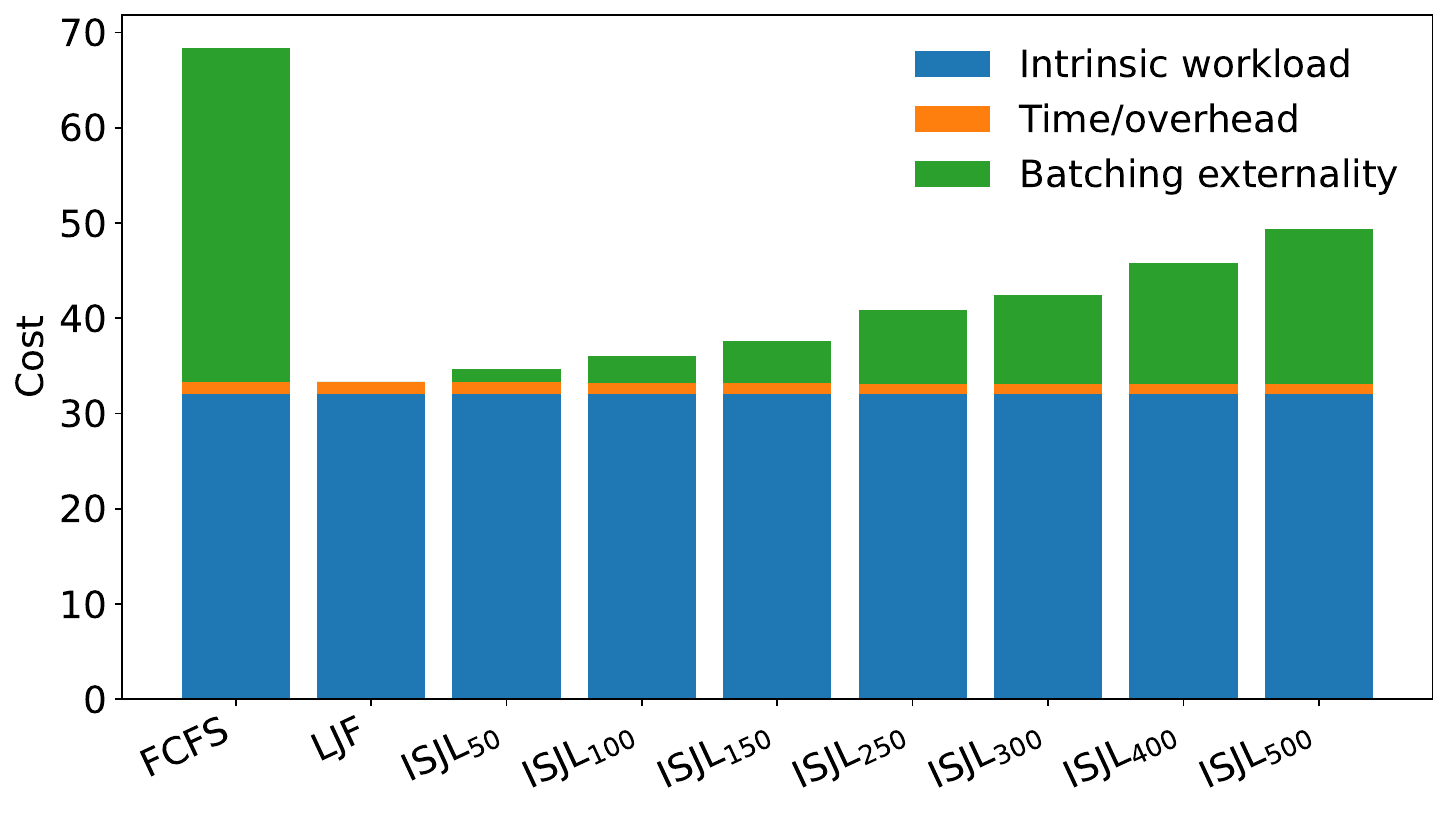}
\includegraphics[width=0.48\textwidth]{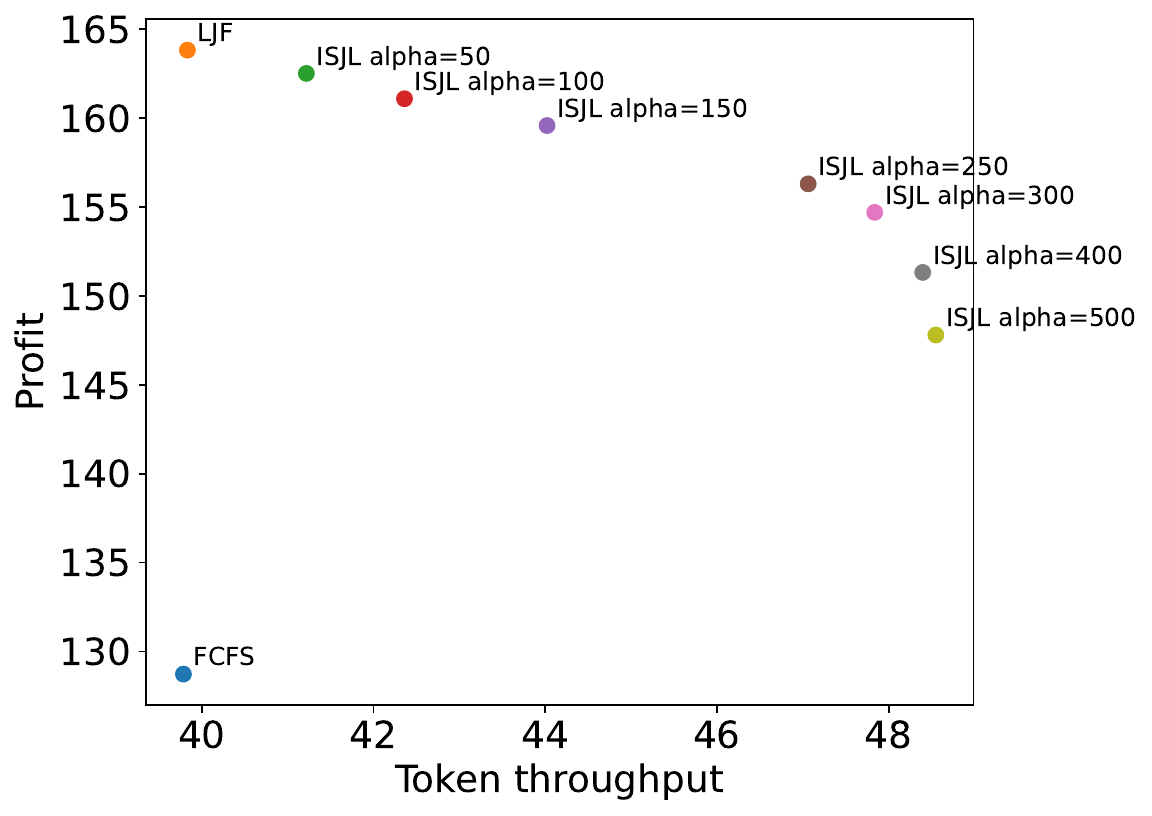}
\caption{Cost decomposition and profit--throughput frontier under linear pricing.}
\label{fig:cost-frontier-numerics}
\end{figure}

Figure~\ref{fig:cost-frontier-numerics} visualizes the same tradeoff. LJF is the
cost-alignment benchmark because its batching externality is zero. FCFS is the
opposite benchmark: it admits requests continuously but allows large progress
heterogeneity, generating high externality cost. ISJL traces an intermediate
frontier. A smaller $\alpha$ keeps ISJL close to LJF in profit, while a larger
$\alpha$ gives more batching flexibility and higher throughput at the cost of a
larger externality term.

\xhdr{When does ISJL dominate LJF in profit?}
The pairwise comparison in Proposition~\ref{prop:pairwise-profit-comparison}
provides a simple way to interpret the fact that LJF has the highest profit in
Table~\ref{tab:cost-decomposition-n400-b50}. Since
$E_{\mathrm{LJF}}(\mathcal I)=0$, ISJL with parameter $\alpha$ has higher profit
than LJF exactly when its time-cost savings exceed its additional externality
cost:
\[
        c\left(T_{\mathrm{LJF}}(\mathcal I)
        -T_{\mathrm{ISJL}(\alpha)}(\mathcal I)\right)
        \ge
        \tau E_{\mathrm{ISJL}(\alpha)}(\mathcal I).
\]
Equivalently, the break-even fixed cost is
\[
        c_\alpha^\star
        =
        \frac{\tau E_{\mathrm{ISJL}(\alpha)}(\mathcal I)}
        {T_{\mathrm{LJF}}(\mathcal I)-T_{\mathrm{ISJL}(\alpha)}(\mathcal I)}.
\]
Table~\ref{tab:break-even-numerics} reports the corresponding empirical
break-even values.

\begin{table}[htbp]
\centering
\caption{Break-even overhead for ISJL relative to LJF}
\label{tab:break-even-numerics}
\begin{tabular}{rrrr}
\toprule
$\alpha$ & Time saving vs. LJF & $\tau E_{\mathrm{ISJL}(\alpha)}$ & $c_\alpha^\star$ \\
\midrule
50 & 83.2 & 1.35 & 0.0162 \\
100 & 147.6 & 2.81 & 0.0190 \\
150 & 235.2 & 4.36 & 0.0185 \\
250 & 380.0 & 7.71 & 0.0203 \\
300 & 414.0 & 9.33 & 0.0225 \\
400 & 437.8 & 12.72 & 0.0291 \\
500 & 444.2 & 16.25 & 0.0366 \\
\bottomrule
\end{tabular}
\end{table}

The baseline value $c=0.0005$ is far below these break-even values. Thus LJF's
profit advantage in Table~\ref{tab:cost-decomposition-n400-b50} is not
surprising: the experiment intentionally places little weight on the number of
service steps. When the fixed per-step overhead becomes economically material,
ISJL's makespan advantage can dominate its controlled externality loss. For
example, ISJL with $\alpha=50$ becomes more profitable than LJF once
$c>0.0162$. This sensitivity result is consistent with the bi-criterion
interpretation of ISJL: it sacrifices a bounded amount of cost alignment to
recover throughput.

\begin{figure}[htbp]
\centering
\includegraphics[width=0.62\textwidth]{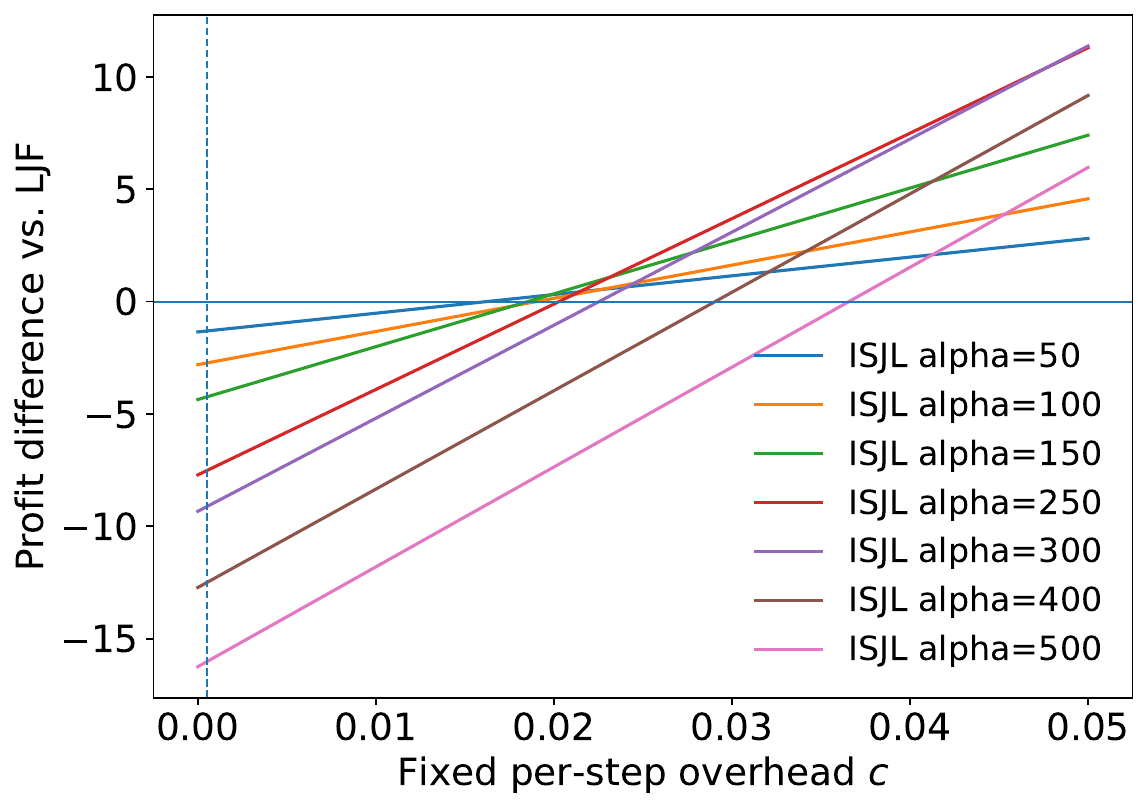}
\caption{Profit difference between ISJL and LJF as the fixed per-step overhead $c$ varies. The dashed vertical line marks the baseline value $c=0.0005$.}
\label{fig:overhead-sensitivity-numerics}
\end{figure}

\subsection{Scope of the Demand Specification and Relation to the Efficiency Benchmark}
\label{subsec:demand-scope}

The profit experiment above should be interpreted as a fixed-workload evaluation
of scheduling policies, not as a model of scheduling-dependent demand. We first
condition on an accepted workload $\mathcal I$ and then hold this workload fixed
across all scheduling policies. Under exogenous token-metered pricing, total
revenue is therefore schedule-independent:
\[
R(\mathcal I)=\rho O(\mathcal I),
\qquad
O(\mathcal I)=\sum_{i\in\mathcal I}o_i .
\]
Consequently, profit differences across FCFS, LJF, and ISJL arise from inference
cost rather than from demand-side revenue changes. In particular, for any two
schedulers $\mathcal A$ and $\mathcal B$ serving the same accepted workload,
\[
\Pi_{\mathcal A}(\mathcal I)-\Pi_{\mathcal B}(\mathcal I)
=
C_{\mathcal B}(\mathcal I)-C_{\mathcal A}(\mathcal I).
\]
This design isolates the operational mechanism analyzed in this paper: how a
scheduler changes the composition of active batches and thereby changes the
schedule-dependent cost components in Theorem~\ref{thm:profit-decomposition},
namely the time/overhead term $cT_{\mathcal A}(\mathcal I)$ and the
batching-externality term $\tau E_{\mathcal A}(\mathcal I)$.

The current analysis therefore does not model a latency-sensitive demand
channel. In practice, scheduling can affect future demand through service
quality: lower waiting time or end-to-end latency may increase user retention,
repeat usage, or willingness to pay, especially for short and interactive
requests. A reduced-form latency-sensitive acceptance model might take the form
\[
v_i-p_i-h_i\mathbb E[W_i(\mathcal A)]\ge 0,
\]
where $v_i$ is the customer's value, $p_i$ is the token-metered price, $h_i$ is
the customer's delay sensitivity, and $W_i(\mathcal A)$ is the waiting time
induced by scheduler $\mathcal A$. Such a model would create a feedback loop
between scheduling, expected waiting time, demand composition, and revenue. It
would also require specifying the arrival process, customer beliefs about
service quality, and the timing of pricing and scheduling decisions. We view
this joint demand--scheduling problem as an important extension, but it is
outside the scope of the present paper.

\section{System Efficiency Under Resource-Fair Scheduling} \label{sec:num}

The previous section quantified cost alignment and profit under exogenous token pricing. This section isolates the operational efficiency of the scheduling policies on a fixed workload from a public dataset, using throughput and average end-to-end latency.

\xhdr{Dataset Overview.}
The experimental data for this study is derived from the LMSYS-Chat-1M dataset \cite{zheng2023lmsys}, a publicly available corpus containing human-chatbot interactions. This dataset was compiled from the Vicuna demo and Chatbot Arena, aggregating conversations from more than 210,000 distinct IP addresses. For the purpose of our simulations, a statistically representative subset of 2,000 conversations was randomly extracted from this larger population.

To model computational demand, we characterize each conversation by its total token count. The sampled subset exhibits considerable diversity in these dimensions. The distribution of total token counts is characterized by a mean of 126, a median of 90, a max of 905, and a variance of 14938.  Figure~\ref{fig:distribution} in Appendix \ref{append:num} illustrates the heterogeneous nature of the workload used in our analysis.

\xhdr{Performance Metrics and Benchmark Algorithms.}
We evaluate the performance of the following scheduling algorithms under batch sizes of $B=16$ and $B=32$:
\begin{itemize}
    \item First-Come, First-Served (FCFS).
    \item Shortest Job First (SJF), which prioritizes requests with the smallest total token count.
    \item Longest Job First (LJF).
    \item The proposed ISJL algorithm with fairness parameters $\alpha=50, 100, 150$.
\end{itemize}

The evaluation is based on two key metrics: throughput and Average End-to-end Latency (AEL). Throughput, defined as the total number of tokens processed per unit time, is calculated as:
\begin{equation} \label{eq:throughput}
    \text{Throughput} = \frac{\text{Total number of tokens processed in }[0, T]}{T},
\end{equation}
where $T$ is the makespan, i.e., the time when all requests have been completed. The AEL is the averaged completion time across all requests:
\begin{equation} \label{eq:ael}
    \text{AEL} = \frac{1}{n} \sum_{i=1}^{n} c_i,
\end{equation}
where $c_i$ is the completion time of request $i$ and $n=2000$ is the total number of requests.

\xhdr{Results.}
The performance results for all algorithms under batch sizes $B=16$ and $B=32$ are summarized in Table~\ref{tab:performanceefficiency}. Our analysis yields several key findings.

{
\xhdr{ISJL dominates all baselines on both metrics 
simultaneously.}
The proposed ISJL algorithm achieves strictly higher 
throughput \emph{and} strictly lower AEL than every 
baseline for all tested values of $\alpha$, across 
both batch sizes. At $B=16$, ISJL with $\alpha=150$ 
achieves the highest throughput ($\mathbf{276}$), 
representing a $14\%$ improvement over LJF ($243$) 
and a $28\%$ improvement over SJF ($196$). 
Simultaneously, ISJL with $\alpha=50$ achieves the 
lowest AEL ($\mathbf{53}$), a $20\%$ improvement 
over LJF ($66$) and a $12\%$ improvement over SJF 
($60$). At $B=32$, the same pattern holds: 
ISJL with $\alpha=150$ achieves throughput of 
$\mathbf{378}$ ($+17\%$ over LJF, $+46\%$ over SJF), 
while ISJL with $\alpha=50$ achieves AEL of 
$\mathbf{62}$ ($-21\%$ over LJF, $-7\%$ over SJF). 
This simultaneous improvement should be interpreted as a comparison with the evaluated benchmark policies, not with the unconstrained optimal schedule. Operationally, ISJL performs well because it combines resource-homogeneous batching with event-driven admissions: it avoids the severe progress imbalance that can arise under unrestricted continuous batching, while also avoiding the idle capacity created by fully static batching.}

\begin{table}[htbp]
\centering
\caption{Throughput and AEL by Different Scheduling Algorithms}
\label{tab:performanceefficiency}
\begin{tabular}{lcccccccc}
\toprule
\multirow{3}{*}{Batch Size} & \multirow{3}{*}{Metrics} & \multicolumn{6}{c}{Scheduling Algorithms} \\
\cmidrule(lr){3-8}
 & & FCFS & SJF & LJF & \multicolumn{3}{c}{ISJL} \\
\cmidrule(lr){6-8}
 & & & & & $\alpha=50$ & $\alpha=100$ & $\alpha=150$ \\
\midrule
\multirow{2}{*}{$B=16$} & Throughput & 216 & 196 & 243 & 262 & 270 & \textbf{276} \\
 & AEL & 62 & 60 & 66 & \textbf{53} & 56 & 54 \\
\cmidrule(lr){1-8}
\multirow{2}{*}{$B=32$} & Throughput & 278 & 259 & 323 & 357 & 369 & \textbf{378} \\
 & AEL & 70 & 67 & 78 & \textbf{62} & 63 & 65 \\
\bottomrule
\end{tabular}
\end{table}

{
\xhdr{The parameter $\alpha$ provides a tunable 
efficiency--latency trade-off.}
Within the ISJL family, the parameter $\alpha$ allows 
operators to tune the balance between throughput and 
AEL. A lower value ($\alpha=50$) enforces stricter 
fairness, reducing the KV-cache spread within each 
batch and thereby minimizing the waiting time for 
individual requests: this configuration achieves 
the best AEL ($53$ at $B=16$, $62$ at $B=32$). 
A higher value ($\alpha=150$) permits greater 
batching flexibility, enabling more requests to be 
processed concurrently and maximizing throughput 
($276$ at $B=16$, $378$ at $B=32$).

}

\section{Conclusion}\label{sec:conclusion}

This paper studies resource-fair scheduling for batched LLM inference. In heterogeneous
continuous batching, short requests can be served in batches whose resource profile is
determined by much longer requests. We formalize this issue through a resource-homogeneity
constraint and design ISJL, a hybrid batching policy that preserves high throughput while
respecting a tunable fairness budget. The theoretical analysis gives a $3/4$ competitive-ratio
guarantee for backlogged instances, and the cost-decomposition analysis shows how the
fairness constraint bounds the schedule-induced batching externality under token-metered
pricing.

Empirically, ISJL occupies a favorable point on the cost--efficiency frontier. Relative to
FCFS, it substantially reduces batching-externality cost; relative to LJF, it recovers much of
the throughput and latency benefit of continuous batching. Thus ISJL should be interpreted
as a bi-criterion scheduling policy: it controls cost misalignment while retaining operational
efficiency.

\paragraph{Discussion: many-server scaling at fixed batch size.}
A natural question is whether the resource-homogeneity problem diminishes in a large-scale
many-server system. In such a regime, a centralized scheduler may
draw from a larger pending pool at each batching epoch, making it easier to form
length-homogeneous cohorts.

The following simple observation formalizes this pooling effect. Suppose that, at a decision
epoch, there are $N$ unstarted pending requests with total token lengths in
$[0,\bar o]$. Fix $\alpha>0$ and partition $[0,\bar o]$ into
$K=\lceil \bar o/\alpha\rceil$ intervals of width at most $\alpha$. Within each interval,
form as many full batches of size $B$ as possible. Since all requests in the same interval
have total lengths differing by at most $\alpha$, each such batch is length-homogeneous.
In interval $k$, at most $B-1$ requests are left over. Therefore, across all intervals, at most
$K(B-1)$ requests are left unmatched, and at least
\[
N-K(B-1)
\]
requests can be assigned to full length-homogeneous batches. Consequently, for fixed
$B$, fixed $\alpha$, and bounded length support, the unmatched fraction is at most
\[
\frac{K(B-1)}{N},
\]
which vanishes as the pending pool size $N$ grows.

This observation supports the intuition that large-scale pooling can mitigate the initial
batch-formation problem. We therefore do not claim that the performance gap between
resource-aware batching and simpler length-clustering policies must persist in every
many-server limit. Instead, our results should be interpreted as finite-instance and
finite-pool guarantees. They are especially relevant when the pending pool is not large
enough to make all batches nearly homogeneous, when workloads are heavy-tailed, when
latency constraints prevent the system from waiting to accumulate a large matching pool,
or when GPU routing and KV-cache locality limit global repacking.

Moreover, the large-pool matching observation concerns the formation of new cohorts of
unstarted requests. It does not by itself remove the dynamic issue created by continuous
batching. Once an active batch has progressed beyond $\alpha$, a newly arrived request
with progress zero cannot be inserted into that batch without violating the fairness
constraint, unless the system allows migration, preemption, or KV-cache repacking across
workers. Thus, large-scale pooling can reduce heterogeneity at admission, but dynamic
admission decisions and finite waiting-time requirements remain operationally important.
A full many-server stochastic analysis with fixed $B$ is an interesting direction for future
work.

\bibliographystyle{ACM-Reference-Format}
\bibliography{reference}

@article{lee2001machine,
  title={Machine scheduling with transportation considerations},
  author={Lee, Chung-Yee and Chen, Zhi-Long},
  journal={Journal of scheduling},
  volume={4},
  number={1},
  pages={3--24},
  year={2001},
  publisher={Wiley Online Library}
}

@article{graham1969bounds,
 ISSN = {00361399},
 URL = {http://www.jstor.org/stable/2099572},
 author = {R. L. Graham},
 journal = {SIAM Journal on Applied Mathematics},
 number = {2},
 pages = {416--429},
 publisher = {Society for Industrial and Applied Mathematics},
 title = {Bounds on Multiprocessing Timing Anomalies},
 urldate = {2025-08-29},
 volume = {17},
 year = {1969}
}

@inproceedings{sheng2024fairness,
  title={Fairness in serving large language models},
  author={Sheng, Ying and Cao, Shiyi and Li, Dacheng and Zhu, Banghua and Li, Zhuohan and Zhuo, Danyang and Gonzalez, Joseph E and Stoica, Ion},
  booktitle={18th USENIX Symposium on Operating Systems Design and Implementation (OSDI 24)},
  pages={965--988},
  year={2024}
}

@article{khan2024ensuring,
  title={Ensuring Fair LLM Serving Amid Diverse Applications},
  author={Khan, Redwan Ibne Seraj and Jain, Kunal and Shen, Haiying and Mallick, Ankur and Parayil, Anjaly and Kulkarni, Anoop and Kofsky, Steve and Choudhary, Pankhuri and Amant, Renee St and Wang, Rujia and others},
  journal={arXiv preprint arXiv:2411.15997},
  year={2024}
}

@article{wei2025equinox,
  title={Equinox: Holistic Fair Scheduling in Serving Large Language Models},
  author={Wei, Zhixiang and Yen, James and Chen, Jingyi and Zhang, Ziyang and Huang, Zhibai and Chen, Chen and Yu, Xingzi and Gu, Yicheng and Wu, Chenggang and Wang, Yun and others},
  journal={arXiv preprint arXiv:2508.16646},
  year={2025}
}

@article{huang2025orlm,
  title={Orlm: A customizable framework in training large models for automated optimization modeling},
  author={Huang, Chenyu and Tang, Zhengyang and Hu, Shixi and Jiang, Ruoqing and Zheng, Xin and Ge, Dongdong and Wang, Benyou and Wang, Zizhuo},
  journal={Operations Research},
  year={2025},
  publisher={INFORMS}
}

@article{wang2025llm,
  title={LLM Serving Optimization with Variable Prefill and Decode Lengths},
  author={Wang, Meixuan and Ye, Yinyu and Zhou, Zijie},
  journal={arXiv preprint arXiv:2508.06133},
  year={2025}
}

@article{chen2025adaptively,
  title={Adaptively Robust LLM Inference Optimization under Prediction Uncertainty},
  author={Chen, Zixi and Ye, Yinyu and Zhou, Zijie},
  journal={arXiv preprint arXiv:2508.14544},
  year={2025}
}

@article{jaillet2025online,
  title={Online Scheduling for LLM Inference with KV Cache Constraints},
  author={Jaillet, Patrick and Jiang, Jiashuo and Mellou, Konstantina and Molinaro, Marco and Podimata, Chara and Zhou, Zijie},
  journal={arXiv preprint arXiv:2502.07115},
  year={2025}
}

@article{ao2025optimizing,
  title={Optimizing LLM Inference: Fluid-Guided Online Scheduling with Memory Constraints},
  author={Ao, Ruicheng and Luo, Gan and Simchi-Levi, David and Wang, Xinshang},
  journal={arXiv preprint arXiv:2504.11320},
  year={2025}
}

@article{shahout2024don,
  title={Don't Stop Me Now: Embedding Based Scheduling for LLMs},
  author={Shahout, Rana and Malach, Eran and Liu, Chunwei and Jiang, Weifan and Yu, Minlan and Mitzenmacher, Michael},
  journal={arXiv preprint arXiv:2410.01035},
  year={2024}
}

@article{mak2015appointment,
  title={Appointment scheduling with limited distributional information},
  author={Mak, Ho-Yin and Rong, Ying and Zhang, Jiawei},
  journal={Management Science},
  volume={61},
  number={2},
  pages={316--334},
  year={2015},
  publisher={INFORMS}
}

@article{kong2013scheduling,
  title={Scheduling arrivals to a stochastic service delivery system using copositive cones},
  author={Kong, Qingxia and Lee, Chung-Yee and Teo, Chung-Piaw and Zheng, Zhichao},
  journal={Operations research},
  volume={61},
  number={3},
  pages={711--726},
  year={2013},
  publisher={INFORMS}
}

@article{allahverdi2008survey,
  title={A survey of scheduling problems with setup times or costs},
  author={Allahverdi, Ali and Ng, Chi To and Cheng, TC Edwin and Kovalyov, Mikhail Y},
  journal={European journal of operational research},
  volume={187},
  number={3},
  pages={985--1032},
  year={2008},
  publisher={Elsevier}
}

@article{chen1998review,
  title={A review of machine scheduling: Complexity, algorithms and approximability},
  author={Chen, Bo and Potts, Chris N and Woeginger, Gerhard J},
  journal={Handbook of Combinatorial Optimization: Volume1--3},
  pages={1493--1641},
  year={1998},
  publisher={Springer}
}

@article{zheng2023lmsys,
  title={{LMSYS-Chat-1M}: A large-scale real-world {LLM} conversation dataset},
  author={Zheng, Lianmin and Chiang, Wei-Lin and Sheng, Ying and Li, Tianle and Zhuang, Siyuan and Wu, Zhanghao and Zhuang, Yonghao and Li, Zhuohan and Lin, Zi and Xing, Eric and others},
  journal={arXiv preprint arXiv:2309.11998},
  year={2023}
}

@article{agrawal2024taming,
  title={Taming throughput-latency tradeoff in {LLM} inference with {Sarathi-Serve}},
  author={Agrawal, Amey and Kedia, Nitin and Panwar, Ashish and Mohan, Jayashree and Kwatra, Nipun and Gulavani, Bhargav S and Tumanov, Alexey and Ramjee, Ramachandran},
  journal={arXiv preprint arXiv:2403.02310},
  year={2024}
}

@article{agrawal2023sarathi,
  title={Sarathi: Efficient {LLM} inference by piggybacking decodes with chunked prefills},
  author={Agrawal, Amey and Panwar, Ashish and Mohan, Jayashree and Kwatra, Nipun and Gulavani, Bhargav S and Ramjee, Ramachandran},
  journal={arXiv preprint arXiv:2308.16369},
  year={2023}
}

@article{patel2023splitwise,
  title={Splitwise: Efficient generative {LLM} inference using phase splitting},
  author={Patel, Pratyush and Choukse, Esha and Zhang, Chaojie and Shah, Aashaka and Goiri, {\'I}{\~n}igo and Maleki, Saeed and Bianchini, Ricardo},
  journal={Power},
  volume={400},
  number={700W},
  pages={1--75},
  year={2023}
}

@article{brown2020language,
  title={Language models are few-shot learners},
  author={Brown, Tom and Mann, Benjamin and Ryder, Nick and Subbiah, Melanie and Kaplan, Jared D and Dhariwal, Prafulla and Neelakantan, Arvind and Shyam, Pranav and Sastry, Girish and Askell, Amanda and others},
  journal={Advances in neural information processing systems},
  volume={33},
  pages={1877--1901},
  year={2020}
}

@article{chowdhery2023palm,
  title={Palm: Scaling language modeling with pathways},
  author={Chowdhery, Aakanksha and Narang, Sharan and Devlin, Jacob and Bosma, Maarten and Mishra, Gaurav and Roberts, Adam and Barham, Paul and Chung, Hyung Won and Sutton, Charles and Gehrmann, Sebastian and others},
  journal={Journal of Machine Learning Research},
  volume={24},
  number={240},
  pages={1--113},
  year={2023}
}

@article{openai2023gpt,
  title={{GPT}-4 technical report. arxiv 2303.08774},
  author={OpenAI},
  journal={View in Article},
  volume={2},
  number={5},
  year={2023}
}

@misc{anthropic2023claude,
  author = {Anthropic},
  title = {Claude},
  howpublished = {\url{https://claude.ai}},
  year = {2023}
}

@misc{chatgpt2023,
  author = {OpenAI},
  title = {{ChatGPT}},
  howpublished = {\url{https://chat.openai.com}},
  year = {2019}
}

@misc{codewhisperer2023,
  author = {{Amazon}},
  title = {Amazon {CodeWhisperer}},
  howpublished = {\url{https://aws.amazon.com/codewhisperer/}},
  year = {2023}
}

@misc{githubcopilot2023,
  author = {{GitHub}},
  title = {{GitHub} Copilot},
  howpublished = {\url{https://github.com/features/copilot}},
  year = {2021}
}

@article{cascella2023evaluating,
  title={Evaluating the feasibility of {ChatGPT} in healthcare: an analysis of multiple clinical and research scenarios},
  author={Cascella, Marco and Montomoli, Jonathan and Bellini, Valentina and Bignami, Elena},
  journal={Journal of medical systems},
  volume={47},
  number={1},
  pages={33},
  year={2023},
  publisher={Springer}
}

@article{sallam2023utility,
  title={The utility of {ChatGPT} as an example of large language models in healthcare education, research and practice: Systematic review on the future perspectives and potential limitations},
  author={Sallam, Malik},
  journal={MedRxiv},
  pages={2023--02},
  year={2023},
  publisher={Cold Spring Harbor Laboratory Press}
}

@inproceedings{kwon2023efficient,
  title={Efficient memory management for large language model serving with {PagedAttention}},
  author={Kwon, Woosuk and Li, Zhuohan and Zhuang, Siyuan and Sheng, Ying and Zheng, Lianmin and Yu, Cody Hao and Gonzalez, Joseph and Zhang, Hao and Stoica, Ion},
  booktitle={Proceedings of the 29th Symposium on Operating Systems Principles},
  pages={611--626},
  year={2023}
}

@inproceedings{yu2022orca,
  title={Orca: A distributed serving system for $\{$Transformer-Based$\}$ generative models},
  author={Yu, Gyeong-In and Jeong, Joo Seong and Kim, Geon-Woo and Kim, Soojeong and Chun, Byung-Gon},
  booktitle={16th USENIX Symposium on Operating Systems Design and Implementation (OSDI 22)},
  pages={521--538},
  year={2022}
}

@article{zhong2024distserve,
  title={Distserve: Disaggregating prefill and decoding for goodput-optimized large language model serving},
  author={Zhong, Yinmin and Liu, Shengyu and Chen, Junda and Hu, Jianbo and Zhu, Yibo and Liu, Xuanzhe and Jin, Xin and Zhang, Hao},
  journal={arXiv preprint arXiv:2401.09670},
  year={2024}
}

@incollection{albers2009online,
  title={Online scheduling},
  author={Albers, Susanne},
  booktitle={Introduction to scheduling},
  pages={71--98},
  year={2009},
  publisher={CRC Press}
}

@inproceedings{precSchedSchabanel,
  author       = {Julien Robert and
                  Nicolas Schabanel},
  editor       = {Shang{-}Hua Teng},
  title        = {Non-clairvoyant scheduling with precedence constraints},
  booktitle    = {Proceedings of the Nineteenth Annual {ACM-SIAM} Symposium on Discrete
                  Algorithms, {SODA} 2008, San Francisco, California, USA, January 20-22,
                  2008},
  pages        = {491--500},
  publisher    = {{SIAM}},
  year         = {2008},
  url          = {http://dl.acm.org/citation.cfm?id=1347082.1347136},
  bibsource    = {dblp computer science bibliography, https://dblp.org}
}

@article{schedPrecedence,
title = {On-line scheduling with precedence constraints},
journal = {Discrete Applied Mathematics},
volume = {119},
number = {1},
pages = {169-180},
year = {2002},
note = {Special Issue devoted to Foundation of Heuristics in Combinatoria l Optimization},
issn = {0166-218X},
doi = {https://doi.org/10.1016/S0166-218X(01)00272-4},
url = {https://www.sciencedirect.com/science/article/pii/S0166218X01002724},
author = {Yossi Azar and Leah Epstein}
}

@inproceedings{precSchedAnupam,
  author       = {Naveen Garg and
                  Anupam Gupta and
                  Amit Kumar and
                  Sahil Singla},
  editor       = {Christel Baier and
                  Ioannis Chatzigiannakis and
                  Paola Flocchini and
                  Stefano Leonardi},
  title        = {Non-Clairvoyant Precedence Constrained Scheduling},
  booktitle    = {46th International Colloquium on Automata, Languages, and Programming,
                  {ICALP} 2019, July 9-12, 2019, Patras, Greece},
  series       = {LIPIcs},
  volume       = {132},
  pages        = {63:1--63:14},
  publisher    = {Schloss Dagstuhl - Leibniz-Zentrum f{\"{u}}r Informatik},
  year         = {2019},
  url          = {https://doi.org/10.4230/LIPIcs.ICALP.2019.63},
  doi          = {10.4230/LIPICS.ICALP.2019.63},
  bibsource    = {dblp computer science bibliography, https://dblp.org}
}

@article{liu2015online,
  title={Online unbounded batch scheduling on parallel machines with delivery times},
  author={Liu, Peihai and Lu, Xiwen},
  journal={Journal of Combinatorial Optimization},
  volume={29},
  pages={228--236},
  year={2015},
  publisher={Springer}
}

@article{li2020online,
  title={Online batch scheduling of simple linear deteriorating jobs with incompatible families},
  author={Li, Wenhua and Wang, Libo and Chai, Xing and Yuan, Hang},
  journal={Mathematics},
  volume={8},
  number={2},
  pages={170},
  year={2020},
  publisher={MDPI}
}

@article{chen2008logistics,
  title={Logistics scheduling with batching and transportation},
  author={Chen, Bo and Lee, Chung-Yee},
  journal={European journal of operational research},
  volume={189},
  number={3},
  pages={871--876},
  year={2008},
  publisher={Elsevier}
}

@online{le2023dissecting,
  author = {Chen, Lequn},
  title = {Dissecting Batching Effects in GPT Inference},
  date = {2023-05-13},
  url = {https://le.qun.ch/en/blog/2023/05/13/transformer-batching/},
  organization = {le.qun.ch}
}

@misc{openai_api_pricing,
  author       = {{OpenAI}},
  title        = {{API Pricing}},
  year         = {2026},
  howpublished = {\url{https://openai.com/api/pricing/}},
  note         = {Accessed: 2026-06-10}
}

@misc{anthropic_api_pricing,
  author       = {{Anthropic}},
  title        = {{Pricing --- Claude API Docs}},
  year         = {2026},
  howpublished = {\url{https://platform.claude.com/docs/en/about-claude/pricing}},
  note         = {Accessed: 2026-06-10}
}

@misc{deepseek_api_pricing,
  author       = {{DeepSeek}},
  title        = {{Models \& Pricing --- DeepSeek API Docs}},
  year         = {2026},
  howpublished = {\url{https://api-docs.deepseek.com/quick_start/pricing}},
  note         = {Accessed: 2026-06-10}
}

@misc{google_gemini_api_pricing,
  author       = {{Google}},
  title        = {{Gemini Developer API Pricing}},
  year         = {2026},
  howpublished = {\url{https://ai.google.dev/gemini-api/docs/pricing}},
  note         = {Accessed: 2026-06-10}
}

@article{golrezaei2023online,
  title={Online resource allocation with convex-set machine-learned advice},
  author={Golrezaei, Negin and Jaillet, Patrick and Zhou, Zijie},
  journal={arXiv preprint arXiv:2306.12282},
  year={2023}
}

@article{lykouris2021competitive,
  title={Competitive caching with machine learned advice},
  author={Lykouris, Thodoris and Vassilvitskii, Sergei},
  journal={Journal of the ACM (JACM)},
  volume={68},
  number={4},
  pages={1--25},
  year={2021},
  publisher={ACM New York, NY}
}

\ECSwitch

\ECHead{Online Appendix}

\section{Supplementary Materials for Section \ref{sec:alg}} \label{append:model}

\subsection{Optimal Scheduling with Integer Programming (IP)}

We index discrete time by a finite horizon $H$ large enough to contain any feasible schedule (e.g., $H=\sum_{i=1}^n o_i$). Let $M$ be a big constant (e.g., $M=\max_i o_i$). Tasks $i\in\{1,\dots,n\}$; time steps $t\in\{1,\dots,H\}$.
Processing requirement $o_i\in\mathbb{Z}_{+}$; parallel capacity $B\in\mathbb{Z}_{+}$;
tolerance $\alpha\in\mathbb{Z}_{+}$.

Decision variables are defined as follows: $u_{i,t}\in\{0,1\}$: task $i$ starts at time $t$; $x_{i,t}\in\{0,1\}$: task $i$ is executing during time $t$; $s_{i,t}\in\mathbb{Z}_{+}$: cumulative progress of task $i$ at the start of time $t$; $T\in\mathbb{Z}_{+}$: makespan (last completion time step). And our objective is defined as:
\[
\min \; T
\]

Constraints are defined as follows: (1) Each task starts exactly once:
\[
\sum_{t=1}^{H-o_i+1} u_{i,t}=1,\quad \forall i.
\]

(2) Start-to-occupancy linkage (contiguous processing for $o_i$ steps once started):
\[
x_{i,t}=\sum_{\tau:\,\tau\le t\le \tau+o_i-1} u_{i,\tau},
\quad \forall i,\ \forall t=1,\dots,H.
\]

(3) Parallel capacity:
\[
\sum_{i=1}^{n} x_{i,t}\le B,\quad \forall t=1,\dots,H.
\]

(4) Progress definition:
\[
s_{i,t}=\sum_{h=1}^{t-1} x_{i,h},\quad \forall i,\ \forall t=1,\dots,H.
\]

(5) Makespan definition:
\[
T \;\ge\; \sum_{t=1}^{H-o_i+1} (t+o_i-1)\,u_{i,t},\quad \forall i.
\]

(6) Fairness Constraint: for any time step, any two executing tasks must satisfy $|s_i-s_j|\le \alpha$, activated only when both tasks execute:
\[
\begin{aligned}
s_{i,t}-s_{j,t} &\le \alpha + M\bigl(2-x_{i,t}-x_{j,t}\bigr),\\
s_{j,t}-s_{i,t} &\le \alpha + M\bigl(2-x_{i,t}-x_{j,t}\bigr),
\end{aligned}
\quad \forall t=1,\dots,H,\ \forall i<j.
\]

\section{Supplementary Materials for Section \ref{sec:alg}} \label{append:alg}

\begin{proof}{Proof of Theorem \ref{thm:robust-fairness}}
The fairness constraint concerns only the progress values $\{a_i(t)\}$ of \emph{currently active} jobs.
Robust-ISJL admits a new job of progress $0$ only when the Phase TWO admission condition is met, i.e.,
when the current maximum progress among active jobs satisfies $A_{\max}(t)\le \alpha$; hence the new
batch has range at most $\alpha$ at the admission time.
Between admissions, all jobs in the active batch advance in lockstep, so the max--min gap is invariant.
When a job completes, it is removed from the active set, and removing elements cannot increase the
max--min range. Therefore the fairness constraint holds at all times.
\end{proof}

\begin{proof}{Proof of Proposition~\ref{prop:online-isjl-fair}.}
New arrivals enter the pending queue and do not affect the active set until they
are admitted. Whenever Online-ISJL invokes Phase ONE, it applies the same packing
rule as Algorithm~\ref{alg:isj-lim-gen} to the currently released queue, so the
active batch constructed by Phase ONE satisfies the $\alpha$-fairness constraint.
Whenever Online-ISJL admits a request during Phase TWO, the request has progress
zero and is admitted only if the current maximum progress among active requests
satisfies $A_{\max}(t)\le \alpha$; hence the progress range after admission is
at most $\alpha$. Between admission epochs, all active requests advance in
lockstep, so the max--min progress gap is invariant. When a request completes,
it is removed from the active set, and removing an element cannot increase the
max--min range. Therefore the fairness constraint holds for all time steps.
\Halmos
\end{proof}

\begin{proof}{Proof of Corollary~\ref{cor:robust-scale-normalized}.}
By Theorem~\ref{thm:robust-cr},
\[
\frac{T_{\mathrm{OPT}}(\mathcal I)}
{T_{\mathrm{Robust\text{-}ISJL}}(\mathcal I)}
\ge
\frac{3}{4}\cdot
\frac{1}{1+\Delta/T_{\mathrm{OPT}}(\mathcal I)}.
\]
Any feasible schedule processes at most $B$ tokens per time step. Therefore,
\[
T_{\mathrm{OPT}}(\mathcal I)\ge \frac{O(\mathcal I)}{B}.
\]
It follows that
\[
\frac{\Delta}{T_{\mathrm{OPT}}(\mathcal I)}
\le
B\frac{\Delta}{O(\mathcal I)}
=
B\bar\varepsilon(\mathcal I).
\]
Substitution gives the first claim. If $u_i\le (1+\varepsilon)o_i$ for all $i$,
then $u_i-o_i\le \varepsilon o_i$ for all $i$, and hence
$\bar\varepsilon(\mathcal I)\le \varepsilon$. The second claim follows.
\Halmos
\end{proof}

\section{Supplementary Materials for Section \ref{sec:extension}} \label{append:extend}

\proof{Proof of Lemma~\ref{lem:gamma-over}}
Let $\alpha\in(0,o_1)$, set $\gamma=\alpha/o_1\in(0,1)$, and normalize $x_i=o_i/o_1$. 
In the Over-Inserting regime, the feasible region is
\[
1\ge x_2\ge x_3\ge x_4>0,\quad 
x_3<1-x_2+\gamma,\quad 
x_3<2\gamma,\quad 
1>\min\{2\gamma,\,x_3+\gamma\},\quad 
x_2+x_3>1,
\]
and the competitive ratio equals
\[
\mathrm{CR}_{\mathrm{over}}=\frac{1+x_3}{x_2+x_3+x_4}.
\]
For fixed $\gamma$, minimizing $\mathrm{CR}$ amounts to maximizing the denominator $x_2+x_3+x_4$ for each fixed $x_3$.

\medskip
\noindent\emph{Maximizing the denominator for fixed $x_3$.}
Since $x_2\ge x_3\ge x_4>0$, the optimal choice is $x_4^\star=x_3$. From
\[
x_2\le 1,\qquad x_3<1-x_2+\gamma\ \Longleftrightarrow\ x_2<1+\gamma-x_3,
\]
together with $x_2\ge x_3$ and $x_2+x_3>1$, the largest admissible $x_2$ is
\[
x_2^\star=\min\{1,\ 1+\gamma-x_3\}.
\]
Hence the maximized denominator for fixed $x_3$ is
\[
D(x_3)=x_2^\star+x_3+x_4^\star=\min\{1,\ 1+\gamma-x_3\}+2x_3=
\begin{cases}
1+2x_3, & x_3\le \gamma,\\
1+\gamma+x_3, & x_3\ge \gamma.
\end{cases}
\]
Define $\Phi(x_3):=(1+x_3)/D(x_3)$; then
\[
\Phi(x_3)=
\begin{cases}
\dfrac{1+x_3}{1+2x_3}, & x_3\le \gamma,\\[6pt]
\dfrac{1+x_3}{1+\gamma+x_3}, & x_3\ge \gamma.
\end{cases}
\]

\medskip
\noindent\emph{Feasible values of $x_3$.}
Besides $x_3>0$, feasibility requires $x_3<2\gamma$ and $1>\min\{2\gamma,x_3+\gamma\}$. This yields:
\[
x_3\in
\begin{cases}
(0,\gamma)\ \cup\ \bigl(\gamma,\ \min\{2\gamma,(1+\gamma)/2\}\bigr), & 0<\gamma<\tfrac12,\\[2pt]
(0,\tfrac12), & \gamma=\tfrac12,\\[2pt]
(0,1-\gamma), & \tfrac12<\gamma<1.
\end{cases}
\]

\medskip
\noindent\emph{One-dimensional minimization.}
\emph{Case $0<\gamma\le\tfrac12$.} On $x_3\le\gamma$,
\[
\Phi'(x_3)=\frac{(1+2x_3)-2(1+x_3)}{(1+2x_3)^2}=-\frac{1}{(1+2x_3)^2}<0,
\]
so $\Phi$ decreases; on $x_3\ge\gamma$,
\[
\Phi'(x_3)=\frac{\gamma}{(1+\gamma+x_3)^2}>0,
\]
so $\Phi$ increases. Therefore the infimum occurs at the junction $x_3=\gamma$:
\[
\inf \mathrm{CR}_{\mathrm{over}}=\frac{1+\gamma}{1+2\gamma}\qquad (0<\gamma\le\tfrac12).
\]

\emph{Case $\tfrac12\le\gamma<1$.} Only the branch $x_3\in(0,1-\gamma)\subseteq(0,\gamma]$ is feasible, where
\[
\Phi(x_3)=\frac{1+x_3}{1+2x_3},\qquad \Phi'(x_3)<0.
\]
Thus the infimum is attained as $x_3\uparrow(1-\gamma)$:
\[
\inf \mathrm{CR}_{\mathrm{over}}=\frac{1+(1-\gamma)}{1+2(1-\gamma)}=\frac{2-\gamma}{3-2\gamma}\qquad (\tfrac12\le\gamma<1).
\]

\medskip
\noindent\emph{Dependence on $\gamma$ and global minimum.}
Let $f_1(\gamma)=(1+\gamma)/(1+2\gamma)$ for $\gamma\in(0,\tfrac12]$ and
$f_2(\gamma)=(2-\gamma)/(3-2\gamma)$ for $\gamma\in[\tfrac12,1)$. Then
\[
f_1'(\gamma)=-\frac{1}{(1+2\gamma)^2}<0,\qquad
f_2'(\gamma)=\frac{1}{(3-2\gamma)^2}>0,
\]
so $f_1$ is strictly decreasing and $f_2$ strictly increasing. At $\gamma=\tfrac12$,
\[
f_1\!\left(\tfrac12\right)=\frac34=f_2\!\left(\tfrac12\right),
\]
hence $\mathrm{CR}_{\mathrm{over}}$ is continuous with a unique global minimum $3/4$. Moreover,
\[
\lim_{\gamma\to0^+}\frac{1+\gamma}{1+2\gamma}=1,\qquad
\lim_{\gamma\to1^-}\frac{2-\gamma}{3-2\gamma}=1.
\]

\medskip
\noindent\emph{Attainability by feasible sequences.}
Because the constraints are strict, the minimizers lie on the boundary but are approached by feasible sequences:
\begin{itemize}
\item If $0<\gamma\le\tfrac12$, take $x_2=1$, $x_4=x_3$, and let $x_3\uparrow\gamma$ with $x_3<\gamma$ to get
$\mathrm{CR}_{\mathrm{over}}=\frac{1+x_3}{1+2x_3}\downarrow \frac{1+\gamma}{1+2\gamma}$.
\item If $\tfrac12\le\gamma<1$, take $x_2=1$, $x_4=x_3$, and let $x_3\uparrow(1-\gamma)$ with $x_3<1-\gamma$ to get
$\mathrm{CR}_{\mathrm{over}}=\frac{1+x_3}{1+2x_3}\downarrow \frac{2-\gamma}{3-2\gamma}$.
\end{itemize}
This completes the proof.
\Halmos 
\endproof

\proof{Proof of Lemma~\ref{lem:gamma-under}}
By scale invariance of $\mathrm{CR}_{\textnormal{under}}$ under a common scaling of $(o_1,o_2,o_3,\alpha)$, normalize $x_i=o_i/o_1$. The feasibility constraints then read
\[
1-x_2+\gamma<x_3\le 2\gamma,
\qquad
1\ge x_2\ge x_3>0,
\]
where the order $x_2\ge x_3$ entails no loss for an adversarial analysis. For
\[
f(x_2,x_3)\coloneqq \frac{x_2+x_3}{1+x_3},
\]
we first observe that for fixed $x_2$,
\[
\frac{\partial f}{\partial x_3}
=\frac{1-x_2}{(1+x_3)^2}\ge 0,
\]
with equality only when $x_2=1$. Hence, for $x_2<1$, the function is strictly \emph{decreasing} in $x_3$; therefore, minimizing $f$ pushes $x_3$ to its largest feasible value, namely
\[
x_3=\min\{\,2\gamma,\ x_2\,\}.
\]
This leads to two regimes, according to which the upper bound is active.

\medskip
\noindent\emph{Regime A: $2\gamma\le x_2$.} Feasibility also requires $1-x_2+\gamma<2\gamma$, i.e.\ $x_2>1-\gamma$. Thus $x_2\in\bigl(\max\{1-\gamma,2\gamma\},1\bigr]$, and
\[
f(x_2,2\gamma)=\frac{x_2+2\gamma}{1+2\gamma}
\]
is strictly increasing in $x_2$. The infimum is attained in the limit at the smallest feasible $x_2$:
\[
\inf f=
\begin{cases}
\dfrac{1+\gamma}{1+2\gamma}, & 0<\gamma\le \tfrac{1}{3}\quad\text{(here $1-\gamma\ge 2\gamma$)},\\[6pt]
\dfrac{4\gamma}{1+2\gamma}, & \tfrac{1}{3}\le \gamma<1\quad\text{(here $1-\gamma\le 2\gamma$)}.
\end{cases}
\]

\medskip
\noindent\emph{Regime B: $x_2\le 2\gamma$.} The lower bound $1-x_2+\gamma<x_3=x_2$ forces $x_2>(1+\gamma)/2$, so $x_2\in\bigl((1+\gamma)/2,\,2\gamma\bigr]$, which is nonempty exactly when $\gamma\ge \tfrac{1}{3}$. With $x_3=x_2$ we have
\[
f(x_2,x_2)=\frac{2x_2}{1+x_2},
\]
strictly increasing in $x_2$, hence minimized in the limit as $x_2\downarrow(1+\gamma)/2$:
\[
\inf f=\frac{2(1+\gamma)}{3+\gamma}.
\]

\medskip
When $0<\gamma\le \tfrac{1}{3}$, Regime~A is the only feasible one, yielding $\mathrm{CR}_{\mathrm{under}}(\gamma)=\dfrac{1+\gamma}{1+2\gamma}$. When $\tfrac{1}{3}\le \gamma<1$, both regimes are feasible; comparing the two values,
\[
\frac{2(1+\gamma)}{3+\gamma}\le \frac{4\gamma}{1+2\gamma}
\iff (1+\gamma)(1+2\gamma)\le 2\gamma(3+\gamma)
\iff \gamma\ge \tfrac{1}{3},
\]
so Regime~B dominates on $[\tfrac{1}{3},1)$. The two expressions match at $\gamma=\tfrac{1}{3}$ with common value $4/5$.
Tightness in each regime follows by taking feasible limits to the identified boundary points.

Finally, the piecewise expression is unimodal in $\gamma$: on $(0,\tfrac{1}{3}]$,
\[
\frac{\mathrm{d}}{\mathrm{d}\gamma}\frac{1+\gamma}{1+2\gamma}
=-\frac{1}{(1+2\gamma)^2}<0,
\]
and on $[\tfrac{1}{3},1)$,
\[
\frac{\mathrm{d}}{\mathrm{d}\gamma}\frac{2(1+\gamma)}{3+\gamma}
=\frac{4}{(3+\gamma)^2}>0.
\]
Thus $\mathrm{CR}_{\mathrm{under}}(\gamma)$ strictly decreases on $(0,\tfrac{1}{3}]$, strictly increases on $[\tfrac{1}{3},1)$, achieves its unique global minimum $4/5$ at $\gamma=\tfrac{1}{3}$, and satisfies
\[
\lim_{\gamma\to 0^+}\frac{1+\gamma}{1+2\gamma}=1,
\qquad
\lim_{\gamma\to 1^-}\frac{2(1+\gamma)}{3+\gamma}=1.
\]
\Halmos 
\endproof

\section{Supplementary Materials for Section \ref{sec:profit-cost-alignment}} \label{append:profit}

\proof{Proof of Theorem~\ref{thm:profit-decomposition}.}
For each service step $t$, write
\[
m^{(t)}:=\max_{j\in S^{(t)}} r_j^{(t)}.
\]
Then
\[
|S^{(t)}|m^{(t)}
=
\sum_{i\in S^{(t)}} m^{(t)}
=
\sum_{i\in S^{(t)}} r_i^{(t)}
+
\sum_{i\in S^{(t)}}\left(m^{(t)}-r_i^{(t)}\right).
\]
Substituting this identity into~\eqref{eq:batch-cost-progress} and summing over
$t=1,\ldots,T_{\mathcal A}(\mathcal I)$ yields
\begin{align*}
C_{\mathcal A}(\mathcal I)
&=
 cT_{\mathcal A}(\mathcal I)
 +\tau\sum_{t=1}^{T_{\mathcal A}(\mathcal I)}\sum_{i\in S^{(t)}} r_i^{(t)}
 +\tau\sum_{t=1}^{T_{\mathcal A}(\mathcal I)}\sum_{i\in S^{(t)}}\left(m^{(t)}-r_i^{(t)}\right).
\end{align*}
The last term is exactly $\tau E_{\mathcal A}(\mathcal I)$. The middle term is
schedule-independent: request $i$ must process token indices
$1,2,\ldots,o_i$ exactly once before completion, regardless of the other
requests with which it is batched. Therefore
\[
\sum_{t=1}^{T_{\mathcal A}(\mathcal I)}\sum_{i\in S^{(t)}} r_i^{(t)}
=
\sum_{i=1}^n\sum_{k=1}^{o_i} k
=
Q(\mathcal I).
\]
Thus
\[
C_{\mathcal A}(\mathcal I)
=
 cT_{\mathcal A}(\mathcal I)+\tau Q(\mathcal I)+\tau E_{\mathcal A}(\mathcal I).
\]
Combining this equality with
$\Pi_{\mathcal A}(\mathcal I)=R(\mathcal I)-C_{\mathcal A}(\mathcal I)$ gives
\eqref{eq:profit-decomposition}.

It remains to prove the externality bound. If $\mathcal A$ is fair, for every active $i\in S^{(t)}$,
\[
0\le m^{(t)}-r_i^{(t)}\le \alpha.
\]
Hence
\[
E_{\mathcal A}(\mathcal I)
\le
\sum_{t=1}^{T_{\mathcal A}(\mathcal I)}\sum_{i\in S^{(t)}}\alpha
=
\alpha\sum_{t=1}^{T_{\mathcal A}(\mathcal I)} |S^{(t)}|.
\]
At each service step, every active request processes exactly one token, so
\[
\sum_{t=1}^{T_{\mathcal A}(\mathcal I)} |S^{(t)}|
=
\sum_{i=1}^n o_i
=
O(\mathcal I).
\]
Therefore $E_{\mathcal A}(\mathcal I)\le \alpha O(\mathcal I)$.
\Halmos
\endproof

\proof{Proof of Theorem~\ref{thm:isjl-profit-guarantee}.}
By Theorem~\ref{thm:profit-decomposition}, revenue and intrinsic resource cost
cancel when comparing two schedules on the same accepted instance. Hence
\begin{align*}
\Pi_\alpha^\star(\mathcal I)-\Pi_{\mathcal A}(\mathcal I)
&=
 c\left(T_{\mathcal A}(\mathcal I)-T_\pi^\star(\mathcal I)\right)
 +\tau\left(E_{\mathcal A}(\mathcal I)-E_{\mathcal A_\pi^\star}(\mathcal I)\right).
\end{align*}
Since $T_{\mathcal OPT}(\mathcal I)\le T_\pi^\star(\mathcal I)$, the competitive
ratio condition~\eqref{eq:rho-makespan-bound} implies
\[
T_{\mathcal A}(\mathcal I)
\le
\frac{1}{\rho}T_{\mathcal OPT}(\mathcal I)
\le
\frac{1}{\rho}T_\pi^\star(\mathcal I),
\]
and therefore
\[
T_{\mathcal A}(\mathcal I)-T_\pi^\star(\mathcal I)
\le
\left(\frac{1}{\rho}-1\right)T_\pi^\star(\mathcal I).
\]
Moreover $E_{\mathcal A_\pi^\star}(\mathcal I)\ge 0$, and because
$\mathcal A$ is fair, Theorem~\ref{thm:profit-decomposition} gives
$E_{\mathcal A}(\mathcal I)\le \alpha O(\mathcal I)$. Thus
\[
\Pi_\alpha^\star(\mathcal I)-\Pi_{\mathcal A}(\mathcal I)
\le
c\left(\frac{1}{\rho}-1\right)T_\pi^\star(\mathcal I)
+
\tau\alpha O(\mathcal I),
\]
which proves~\eqref{eq:generic-profit-guarantee}. Setting
$\rho=3/4$ for ISJL by Theorem~\ref{thm:cr-isj-lim-gen} gives
\eqref{eq:isjl-profit-guarantee}.
\Halmos
\endproof

\proof{Proof of Proposition~\ref{prop:pairwise-profit-comparison}.}
Equation~\eqref{eq:pairwise-profit-identity} follows immediately by subtracting
the decompositions in~\eqref{eq:profit-decomposition} for $\mathcal A$ and
$\mathcal B$; the fixed terms $R(\mathcal I)$ and $\tau Q(\mathcal I)$ cancel.

Under LJF, each static batch starts with equal progress on all active requests,
and no new request is admitted into that batch before the batch is rebuilt. While
requests in a batch are active, their progress values are equal; when shorter
requests complete, removing them from the active set cannot create a positive
max--min gap among the remaining active requests. Thus
$E_{\mathcal LJF}(\mathcal I)=0$. Taking
$(\mathcal A,\mathcal B)=(\mathcal ISJL,\mathcal LJF)$ in
\eqref{eq:pairwise-profit-identity} and using
$E_{\mathcal ISJL}(\mathcal I)\le \alpha O(\mathcal I)$ gives
\eqref{eq:isjl-ljf-profit-comparison}. Taking
$(\mathcal A,\mathcal B)=(\mathcal ISJL,\mathcal FCFS)$ and again using
$E_{\mathcal ISJL}(\mathcal I)\le \alpha O(\mathcal I)$ gives
\eqref{eq:isjl-fcfs-profit-comparison}.
\Halmos
\endproof

\proof{Proof of Proposition~\ref{prop:fcfs-quadratic-externality}.}
Under FCFS, the long request occupies one server for $L$ service steps. Since
there are exactly $L$ short requests and each has length $1$, FCFS places one
short request on the other server in each of these $L$ steps. Hence all work is
completed when the long request completes, so $T_{\mathcal FCFS}(\mathcal I)=L$.
At the $k$th service step, $k=1,\ldots,L$, the long request processes token index
$k$, while the short request processes token index $1$. Therefore the
externality in step $k$ is
\[
(k-k)+(k-1)=k-1.
\]
Summing over $k=1,\ldots,L$ gives
\[
E_{\mathcal FCFS}(\mathcal I)=\sum_{k=1}^L (k-1)=\frac{L(L-1)}{2}.
\]

Under LJF, the long request and one short request are placed in the first static
batch. They have equal progress during the first step; after the short request
completes, the long request runs alone until completion. The remaining $L-1$
short requests are then processed in batches of size two, requiring
$\lceil(L-1)/2\rceil$ additional steps. Thus
$T_{\mathcal LJF}(\mathcal I)=L+\lceil(L-1)/2\rceil$. Since every LJF batch starts
with equal progress and admits no new request while the batch is active,
$E_{\mathcal LJF}(\mathcal I)=0$. Finally,
\eqref{eq:ljf-fcfs-profit-adversarial} follows from the pairwise comparison
identity~\eqref{eq:pairwise-profit-identity}. Because the positive term in
\eqref{eq:ljf-fcfs-profit-adversarial} grows quadratically in $L$, whereas the
negative term grows only linearly, LJF is more profitable than FCFS for all
sufficiently large $L$.
\Halmos
\endproof

\proof{Proof of Corollary \ref{cor:makespan-certified-profit}}
By Theorem~\ref{thm:profit-decomposition}, the profit of schedule
$\mathcal A$ can be written as
\[
\Pi_{\mathcal A}(\mathcal I)
=
R(\mathcal I)-\tau Q(\mathcal I)
-cT_{\mathcal A}(\mathcal I)
-\tau E_{\mathcal A}(\mathcal I).
\]
If $\mathcal A$ is $\alpha$-fair, then
\[
E_{\mathcal A}(\mathcal I)\le \alpha O(\mathcal I).
\]
Substituting this bound into the profit decomposition gives
\[
\Pi_{\mathcal A}(\mathcal I)
\ge
R(\mathcal I)-\tau Q(\mathcal I)-\tau \alpha O(\mathcal I)
-cT_{\mathcal A}(\mathcal I).
\]
For a fixed workload $\mathcal I$ and fixed $\alpha$, the first three terms on
the right-hand side are independent of the schedule. Therefore, maximizing the
certified profit lower bound over $\alpha$-fair schedules is equivalent to
minimizing $T_{\mathcal A}(\mathcal I)$.
\Halmos
\endproof

\proof{Proof of Proposition \ref{prop:nested-alpha-frontier}}
Any schedule satisfying the $\alpha_1$-fairness constraint has within-batch
progress gap at most $\alpha_1$ at every time step. Since
$\alpha_1\le \alpha_2$, the same schedule also satisfies the
$\alpha_2$-fairness constraint. Hence
$\mathcal F_{\alpha_1}\subseteq \mathcal F_{\alpha_2}$. Minimizing makespan
over the larger feasible set cannot yield a larger optimum, so
$T_{\alpha_2}^\star(\mathcal I)\le T_{\alpha_1}^\star(\mathcal I)$. Since
$O(\mathcal I)$ is fixed for the instance, throughput
$O(\mathcal I)/T_\alpha^\star(\mathcal I)$ is weakly increasing in $\alpha$.
\Halmos
\endproof

\section{Supplementary Materials for Section \ref{sec:num}}\label{append:num}

Figure \ref{fig:distribution} displays the distribution of the number of tokens.

\begin{figure}[!ht]
\centering
\includegraphics[width=0.8\textwidth]{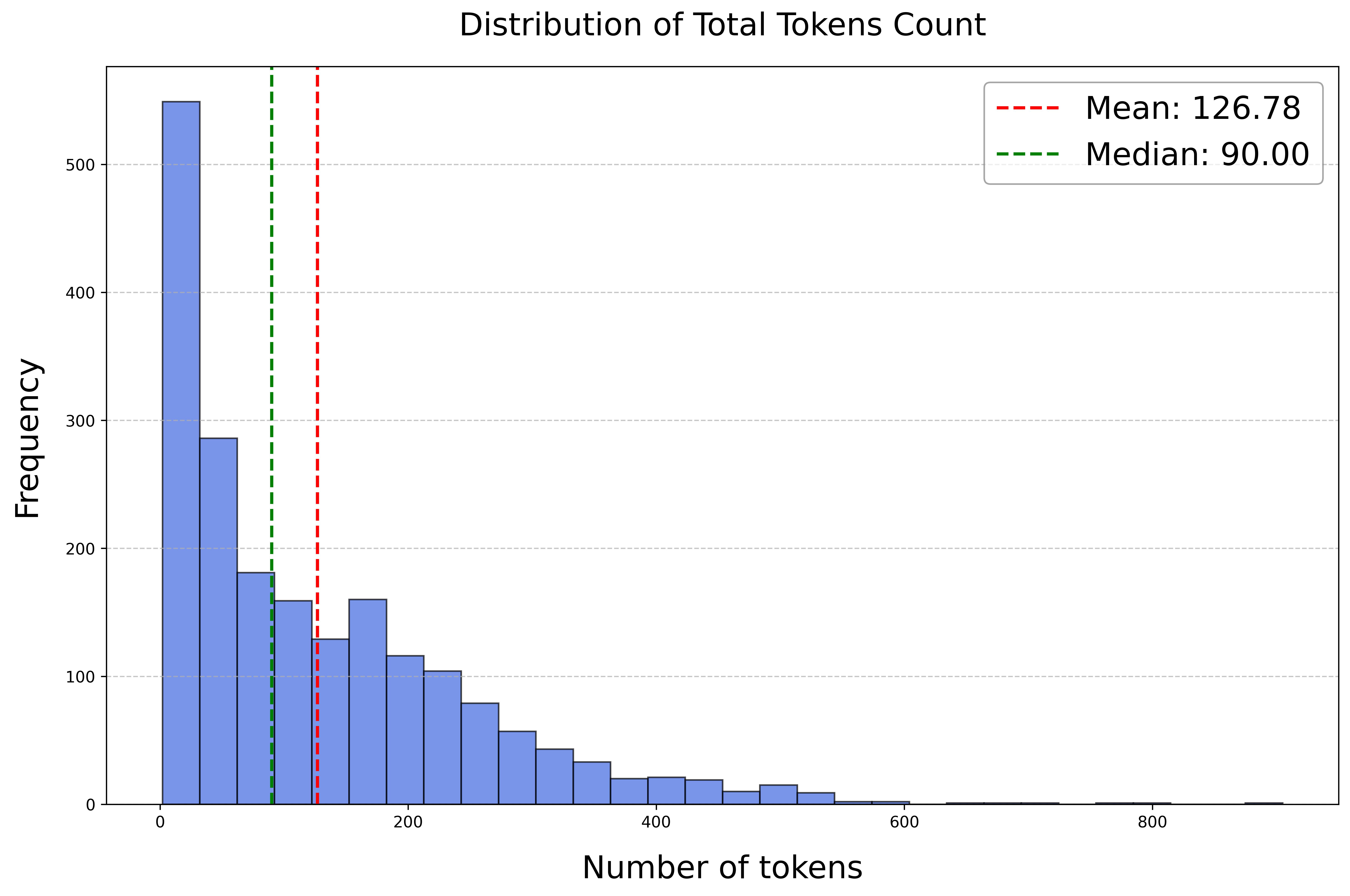}
\caption{Distribution of the Number of Tokens} 
\label{fig:distribution}
\end{figure}

%
%
%

\end{document}